\documentclass[aps,pra,twocolumn,superscriptaddress,10pt]{revtex4-2}
\usepackage{amsthm}
\usepackage{amsmath,bm}
\usepackage{amssymb}
\usepackage{amsfonts}
\usepackage{graphicx}
\usepackage{txfonts}
\usepackage{xcolor}
\usepackage{float}
\usepackage{braket}
\usepackage[colorlinks=true,linkcolor=blue,citecolor=blue,urlcolor=blue]{hyperref}
\usepackage[capitalise]{cleveref}
\usepackage{bbm}
\usepackage{changes}
\usepackage{ragged2e}
\usepackage[justification=justified, format=plain]{subcaption}
\usepackage[justification=raggedright]{caption}

\newcommand{\Tr}{\text{Tr}}

\newcommand{\ignore}[1]{} 

\newcounter{SaveEqnCntr}

\newcommand{\be}{\begin{align}}
\newcommand{\ee}{\end{align}}
\newcommand{\ba}{\begin{eqnarray}}
\newcommand{\ea}{\end{eqnarray}}

\newtheorem{theorem}{Theorem}
\newtheorem{corollary}{Corollary}
\newtheorem{definition}{Definition}
\newtheorem{proposition}{Proposition}

\newtheorem{remark}{Remark}
\newtheorem{lemma}{Lemma}

\def\>{\rangle}
\def\<{\langle}

\begin{document}
	
\title{Measurement-induced remote activation of nonclassicality}

\author{Sudip Chakrabarty}
\email{sudip27042000@gmail.com}
\affiliation{S. N. Bose National Centre for Basic Sciences, Block JD, Sector III, Salt Lake, Kolkata 700 106, India}

\begin{abstract}
Nonclassicality, defined through nonpositivity of quasiprobability distributions such as the Wigner and Kirkwood-Dirac distributions, underlies key quantum advantages. We show that it can be remotely activated: local measurements by a correlated partner alone can drive a party whose reduced state is classical into a genuinely nonclassical conditional state. For two-qubit systems we obtain closed forms for the maximal single-branch effect and for its probability-weighted average under the best two-outcome projective measurement, the latter proven optimal over all measurements when the shared state carries no local bias. Local noise on the steered party, evaluated relative to its own noise-updated reference basis, can enhance the effect when non-unital and misaligned with that party's axis, whereas no unital channel ever helps. This nonclassicality is shown to unlock, for the same reference observable, anomalous weak values unavailable to the unsteered state under arbitrary postselection. An exactly solvable hybrid qubit-oscillator model shows the same mechanism generates Wigner negativity. These results identify measurement-induced steering as an operational resource for remotely activating nonclassicality.
\end{abstract}

\maketitle

\section{Introduction} \label{sec:intro}
 
Whether a quantum state or process admits any classical description is ultimately a question about the mathematical objects available to represent it. Superposition~\cite{streltsov2017colloquium}, entanglement~\cite{horodecki2009quantum}, nonlocality~\cite{RevModPhys.86.419}, and contextuality~\cite{budroni2022kochen,PhysRevA.71.052108} name the specific ways a quantum state or measurement can defeat a classical account, each underlying part of the advantage quantum systems offer for computation, communication, and sensing. Quasiprobability distributions make this boundary formally sharp: they represent a state through a function or array that behaves like an ordinary probability distribution in every respect but one, its values allowed to go negative or complex, a signature with no classical counterpart~\cite{ferrie2011quasi,cahill1969density,tan2020negativity,veitch2012negative}.

Prominent examples include the Wigner function~\cite{wigner1932quantum,kenfack2004negativity}, the Glauber-Sudarshan $P$ function~\cite{glauber1963coherent,sudarshan1963equivalence}, and the Husimi $Q$ function~\cite{husimi1940some}. In continuous-variable (CV) systems, Wigner negativity is widely regarded as a hallmark of nonclassicality~\cite{kenfack2004negativity, mari2011directly, mallick2025efficient, chakrabarty2026operationaldetectionwignernegativity}, and, beyond its foundational significance, is a central resource for quantum advantage in quantum information and computation, notably through its established equivalence with quantum contextuality~\cite{PhysRevLett.129.230401} and its role as an obstruction to classical simulability~\cite{mari2012positive}.

A complementary quasiprobability framework, better suited to discrete-variable (DV) systems, is the Kirkwood-Dirac (KD) distribution~\cite{arvidsson2024properties,PhysRevA.109.012215,budiyono2023quantifying,umekawa2024advantages, chakrabarty2026probingkirkwooddirac}, introduced by Kirkwood~\cite{kirkwood1933quantum} and Dirac~\cite{dirac1945analogy}: it assigns complex quasiprobabilities to pairs of measurement bases, capturing both statistical and dynamical aspects of a quantum system, its real part, the Margenau-Hill distribution~\cite{margenau1961correlation}, encoding nonclassical features in its own right. KD negativity or non-reality has been linked to contextuality~\cite{kunjwal2019anomalous,pusey2014anomalous}, quantum computation~\cite{x819-898d}, measurement back-action~\cite{jozsa2007complex,dressel2012significance,monroe2021weak}, thermodynamic anomalies~\cite{lostaglio2020certifying,gonzalez2019out}, quantum imaginarity~\cite{chakrabarty2026detectionquantumimaginarityusing}, and weak-value amplification~\cite{arvidsson2020quantum}, and is itself accessible experimentally through weak measurements, interferometry, and direct reconstruction~\cite{PhysRevA.76.012119,PhysRevLett.110.230602,PhysRevLett.110.230601,lundeen2011direct,PhysRevLett.108.070402}, typically quantified by the total weight of negative or non-real terms~\cite{gonzalez2019out,PhysRevA.109.012215,PhysRevLett.127.190404,arvidsson2021conditions,PhysRevA.106.022217,de2023relating,Razieh19}.

With the growing use of nonclassical elements in quantum information processing, we ask whether nonclassicality of either kind, Wigner negativity in CV systems or KD nonpositivity in DV systems, can be activated operationally, on demand, through local actions on a correlated system, rather than treated as a fixed property a state either possesses or lacks. We address this through the lens of quantum steering. Steering is most often discussed as the absence of local-hidden-state (LHS) models~\cite{wiseman2007steering}, asking whether the assemblage of conditional states Alice can prepare for Bob admits a classical explanation; that is not the question asked here. We instead use steering in a purely geometric, operational sense, first made precise for two-qubit states through the quantum steering ellipsoid~\cite{jevtic2014quantum}: which states Alice's local measurements can actually prepare for Bob, regardless of whether the resulting assemblage could also be explained by an LHS model.

A related terminological point is worth flagging here, developed fully in Sec.~\ref{s2A}. That a classical, resource-free $\rho_B$ can admit an ensemble decomposition containing individually nonclassical pure states, and that such a decomposition can always be remotely realized through a suitable measurement on a purifying system, is guaranteed in general by the Hughston-Jozsa-Wootters (HJW) theorem and is the same mechanism behind entanglement of assistance and coherence of assistance; this bare possibility is accordingly not new to this work. What this paper contributes beyond it is an exact two-qubit optimization of the resulting KD nonclassicality over both Alice's measurement and Bob's evaluation basis, a closed-form geometric criterion for exactly when the effect vanishes, a complete account of how local noise on Bob's side affects it, and explicit realizations of the resource as anomalous weak values and as steered Wigner negativity.

For two-qubit KD nonclassicality, we obtain the maximal quantity in closed form for an arbitrary shared state, and the average quantity in closed form for the best two-outcome projective measurement; the two provably coincide, making the latter fully optimal, whenever the shared state carries no local Bloch bias, and otherwise the second is only a proven lower bound on the fully optimized average. We give an exact algebraic vanishing condition for both quantities, and prove that no unital channel on Bob, short of erasing his own polarization direction, can increase either one, while an explicit non-unital amplitude-damping channel can, once sufficiently misaligned with that axis past a numerically located threshold. We connect the KD resource to anomalous weak values, proving that, for the reference observable singled out by this construction, a party holding only the unsteered state is denied this capacity under arbitrary postselection while steering restores it. The same construction, carried into the CV setting, steers a Wigner-positive marginal into genuinely Wigner-negative conditional states in an exactly solvable qubit-oscillator model, with the negativity itself reduced, at fixed measurement settings, to an exact one-dimensional integrand, though the optimal setting and resulting maximal and average negativities are located numerically.

The remainder of this paper is organized as follows. Section~\ref{sec:preliminaries} reviews the geometric notion of quantum steering used throughout, together with the KD and Wigner quasiprobability constructions. Section~\ref{sec:properties} develops the general framework, defining the maximal and average steered nonclassicality, establishing their convexity-based properties, and identifying the states incapable of remote activation and those that saturate it. Section~\ref{sec:twoqubitKD} specializes this framework to two-qubit systems and KD nonclassicality, obtaining both quantities in closed form and illustrating them on worked examples. Section~\ref{sec:local_ops} examines how local noise on the steered party affects the result. Section~\ref{sec:application} connects the resulting resource to anomalous weak values. Section~\ref{sec:steered_wigner} carries the construction into the CV setting and studies steered Wigner negativity. Section~\ref{s_conclusions} concludes with a summary and open directions; the longer proofs and the numerical procedure are collected in the Appendix.

\section{Preliminaries}\label{sec:preliminaries}

\subsection{Quantum steering and the steering ellipsoid}\label{s2A}
 
Quantum steering was first noted by Schr\"odinger \cite{schrodinger1935discussion} in response to the Einstein-Podolsky-Rosen (EPR) paradox \cite{einstein1935can}: local measurements by Alice can remotely prepare, at Bob's location, an ensemble of conditional states depending on her choice of measurement, with no physical interaction between them. One programme treats this operationally: a state is steerable if this assemblage cannot be reproduced by Bob holding a fixed ensemble about which Alice has only classical information, a local-hidden-state (LHS) model~\cite{wiseman2007steering}, reviewed in \cite{cavalcanti2017quantum,uola2020quantum}. This notion plays no role here: we never ask whether $\rho_{AB}$ is steerable in this sense. Accordingly, whenever the term ``steering" is used below, it refers to this measurement-induced remote state preparation, or equivalently to the geometry of the set of states reachable at Bob by Alice's local measurements, and not to EPR steering as certified by the failure of an LHS model.
 
This work builds on a second, geometric programme: steering as a statement about which states at Bob's location are \emph{reachable} by Alice's measurements, regardless of whether the same set could also arise from an LHS strategy. For two qubits, the Bloch vectors one party can remotely prepare at the other's location, over all POVMs, trace out an ellipsoid~\cite{jevtic2014quantum}, the quantum steering ellipsoid (QSE), inside the second party's Bloch ball, its centre, orientation, and semi-axes fixed by $\rho_{AB}$'s local Bloch vectors and correlation matrix. With Alice measuring and Bob steered, this is the set traced by $\boldsymbol m\mapsto\boldsymbol r(\boldsymbol m)$ as Alice's direction ranges over the Bloch sphere; it is this reachable-set object that underlies the constructions of Sec.~\ref{sec:properties} and \ref{sec:twoqubitKD}.
 
The same reachable-set idea extends beyond a qubit Bob. Whenever Alice's system is a qubit, Corollary~\ref{cor:rankone} lets her POVM be restricted to rank-one elements without loss of generality, so the states Bob can be steered to are the image of a single Bloch sphere of measurement directions under a fixed map into Bob's state space. For two qubits this image is the ellipsoid above; in Sec.~\ref{sec:steered_wigner}, where Bob instead holds a CV mode, the same construction gives an explicit one-parameter family of coherent-state superpositions, Eq.~\eqref{eq:psi_pm}, again the image of Alice's Bloch sphere, now landing in an infinite-dimensional Hilbert space rather than a finite ellipsoid. The steering ellipsoid is thus the two-qubit case of one reachable-set construction that reappears throughout Secs.~\ref{sec:twoqubitKD}-\ref{sec:steered_wigner}.
 
A terminological point is worth making explicit. The nonclassicality reported below is not manufactured from nothing: it is present, operationally, in the correlations of $\rho_{AB}$ from the outset, and Alice's measurement never acts on Bob's system at all; it only selects, and communicates to Bob, which member of an ensemble decomposition of $\rho_B$ he in fact holds. That such a decomposition exists whenever it is asked to contain a chosen target pure state, and can always be remotely realized through a suitable POVM on Alice's side, is guaranteed in general by the HJW theorem~\cite{hughston1993complete}, invoked explicitly in the proof of Theorem~\ref{thm:optimal} below, and is structurally the same observation underlying assisted-resource constructions such as entanglement of assistance and coherence of assistance, a precedent already invoked for $\mathcal C_{\mathcal N}^{\mathrm a}(\rho_B)$ after Theorem~\ref{thm:optimal}. ``Activation," as used throughout, refers to this act of conditioning, turning an already-classical average into an accessible, individually nonclassical branch, not to any local process creating nonclassicality the shared state's correlations could not already support; as already flagged in Sec.~\ref{sec:intro}, the distinctive content of this work lies not in that general, HJW-guaranteed possibility but in the exact two-qubit optimization, the steering-geometrical characterization of when the effect vanishes, and the analysis of how local channels on Bob affect it, developed below.

\subsection{Nonclassicality via quasiprobability nonpositivity}\label{s2B}

We now review the two quasiprobability representations used throughout this work: the KD distribution, suited to finite-dimensional systems, and the Wigner function, suited to CV systems.

\subsubsection{KD quasiprobability distribution}\label{s2B1}

We consider a Hilbert space of dimension $d$, on which all relevant operators are defined. Let $\{\ket{a_i}\}$ and $\{\ket{b_k}\}$ be two orthonormal basis sets associated with the spectral resolutions of the observables
\(
A=\sum_i a_i\,\ket{a_i}\bra{a_i},\;
B=\sum_k b_k\,\ket{b_k}\bra{b_k}.
\)
Given a density operator $\rho$, the KD quasiprobability distribution corresponding to these two bases is defined as
\begin{align}
Q_{ik}(\rho)
\equiv \langle b_k|a_i\rangle\,\bra{a_i}\rho\ket{b_k}
= \mathrm{Tr}\!\left(\Pi^b_k \Pi^a_i \rho\right),
\label{Eq:KD}
\end{align}
with $\Pi^a_i=\ket{a_i}\bra{a_i}$ and $\Pi^b_k=\ket{b_k}\bra{b_k}$. 
The quantities $Q_{ik}(\rho)$ satisfy basic consistency relations,
\begin{align}
\sum_{i,k} Q_{ik} = 1, \quad \sum_k Q_{ik} = \mathrm{Tr}(\Pi^a_i \rho), \quad \sum_i Q_{ik} = \mathrm{Tr}(\Pi^b_k \rho),
\end{align}
ensuring that the marginals reproduce the Born rule probabilities for measurements in the $\{\ket{a_i}\}$ and $\{\ket{b_k}\}$ bases.
When the two bases coincide, $Q_{ik}(\rho)=\mathrm{Tr}(\Pi^a_i \rho)\,\delta_{ik}$, an ordinary probability distribution. For general bases and states, however, the KD elements need not be real or positive and therefore, in general, cannot be interpreted as an ordinary probability distribution; the total weight of negative and non-real elements is the resource of interest here, and we refer to it throughout as KD sum-negativity, following Eq.~\eqref{eq:N_measure} below, to keep in view that it jointly captures sign negativity and non-reality rather than sign alone.

The KD framework extends to a sequence of $l$ observables $A^{(r)}=\sum_i a^{(r)}_i\Pi^{a^{(r)}}_i$, $r=1,\ldots,l$, via the extended distribution
\begin{align}
Q^\star_{i_1,\ldots,i_l}(\rho)
=
\mathrm{Tr}\!\left(
\Pi^{a^{(l)}}_{i_l}\cdots
\Pi^{a^{(1)}}_{i_1}\,\rho
\right),
\label{KDExt}
\end{align}
which encodes complete information about $\rho$: provided the relevant overlaps are nonzero, the state is reconstructed as
\begin{align}
\rho
=
\sum_{i_1,\ldots,i_l}
\frac{
\ket{a^{(1)}_{i_1}}\bra{a^{(l)}_{i_l}}
}{
\langle a^{(l)}_{i_l}|a^{(1)}_{i_1}\rangle
}
\,Q^\star_{i_1,\ldots,i_l}(\rho).
\label{KDDecom}
\end{align}
KD quasiprobabilities have found applications throughout quantum information, foundations, and thermodynamics, where negativity or non-reality is used to witness contextuality, quantify measurement disturbance, and constrain thermodynamic and metrological protocols \cite{Dressel15,Yunger18,gonzalez2019out,Razieh19,arvidsson2020quantum,kunjwal2019anomalous,pusey2014anomalous,dressel2011experimental,lostaglio2020certifying}; see \cite{arvidsson2024properties} for a detailed review.

\subsubsection{Wigner function}\label{s2B2}

The Wigner function is the corresponding quasiprobability representation for CV systems, providing a complete phase-space description of a quantum state \cite{cahill1969density}. Writing $(x_k,p_k)$ for the position and momentum quadratures of the $k$-th mode, with $\vec x=(x_1,\ldots,x_N)$, $\vec p=(p_1,\ldots,p_N)$, the Wigner function of an $N$-mode density matrix $\hat\rho$ is
\begin{align}\label{eq:wigner_def}
W(\vec{x}, \vec{p}) = \frac{1}{{(2 \pi) }^N} \int\limits_{-\infty}^{\infty} \bra{ \vec{x} +\frac{\vec{y}}{2}}  \hat{\rho} \ket{\vec{x} - \frac{\vec{y}}{2} } e^{-i \vec{p} \cdot \vec{y}} \, d^N y,
\end{align}
which for a single mode reduces to \cite{wigner1932quantum}
\begin{align}
  W(x,p)=   \frac{1}{2 \pi} \int\limits_{ -\infty}^{\infty} \bra{x+\frac{y}{2}}\hat{\rho} \ket{x-\frac{y}{2}}e^{-ipy} \, dy,
\end{align}
and, for a pure state $\ket\psi$ with wavefunction $\psi(x)$, to
\begin{align}
W(x,p)=\frac{1}{2\pi} \int\limits_{-\infty}^{\infty}  \psi\!\left(x+\frac{y}{2}\right)\psi^{*}\!\left(x-\frac{y}{2}\right) e^{-ipy} \, dy.
\end{align}
The Wigner function is real-valued by construction, and normalized so that $\int\!\int W(x,p)\,dx\,dp = 1$. It resembles a probability distribution but can take negative values; its marginals, by contrast, are genuine probability distributions, realizable by homodyne detection,
\begin{align}
|\psi(x)|^2 = \int\limits_{-\infty}^{\infty} W(x,p)\,dp,\quad |\phi(p)|^2 = \int \limits_{-\infty}^{\infty} W(x,p)\,dx,
\end{align}
with $\phi(p)$ the momentum wavefunction.

CV states split into two classes: Gaussian states, whose Wigner function is Gaussian and hence positive, and non-Gaussian states. By Hudson's theorem \cite{hudson1974wigner}, a \emph{pure} state with positive Wigner function is always Gaussian, but this fails for mixed states, for which Wigner-positive non-Gaussian states exist. Wigner negativity is accordingly taken as the figure of merit for CV nonclassicality throughout this work.

\subsubsection{Quantifying nonclassicality}
Throughout, we make use of a structural property common to both notions of nonclassicality considered in this work.

\begin{lemma}
\label{lem:convexity}
The KD quasiprobability $Q(\rho)$ and the Wigner function $W(\rho)$ are linear functionals of $\rho$. Consequently, the associated nonclassicality measure
\begin{equation} \label{eq:N_measure}
\mathcal N(\rho) = \tfrac12\big(\|\,\text{quasiprobability of }\rho\,\|_1 - 1\big)
\end{equation}
is a convex function of $\rho$, and the classical set $\mathcal C = \{\rho : \mathcal N(\rho) = 0\}$ is convex.
\end{lemma}

\begin{proof}
Linearity of $Q$ and $W$ in $\rho$ is immediate from their definitions. The modulus of a linear functional of $\rho$ is a convex function of $\rho$, and a non-negative sum or integral of convex functions is convex; hence $\|\,\text{quasiprobability of }\rho\,\|_1$, and therefore $\mathcal N(\rho)$, is convex. For the second claim, if $\mathcal N(\rho_1) = \mathcal N(\rho_2) = 0$, then for any $q \in [0,1]$, convexity of $\mathcal N$ together with $\mathcal N \geq 0$ gives
\[
0 \leq \mathcal N(q\rho_1 + (1-q)\rho_2) \leq q\mathcal N(\rho_1) + (1-q)\mathcal N(\rho_2) = 0,
\]
so $q\rho_1 + (1-q)\rho_2 \in \mathcal C$.
\end{proof}

A point worth making explicit: linearity of $Q(\rho)$, and hence convexity of $\mathcal N_{KD}$, holds for a \emph{fixed} pair of KD bases; optimizing over the second basis afterward preserves convexity, since a pointwise maximum of convex functions is itself convex. Every application of Lemma~\ref{lem:convexity} in this work, including in the proof of Theorem~\ref{thm:incapable} below, holds the first (reference) KD basis fixed throughout the argument, at the value determined once by the state under consideration at the outset, before any subsequent mixing over conditional branches; the lemma is never applied with a reference basis that is simultaneously varied as a function of the state being mixed.

This lemma is the mechanism behind every result that follows: it is precisely because the classical set is convex that mixing classical conditional states can never produce a nonclassical average, while conversely a nonclassical average can only arise from an underlying mixture in which at least one conditional state was already nonclassical.

\section{Framework}
\label{sec:properties}

Let $\rho_{AB}$ be a bipartite state shared by Alice and Bob, with Bob's reduced state $\rho_B = \Tr_A(\rho_{AB})$ satisfying $\mathcal N(\rho_B) = 0$ 
for a chosen nonclassicality functional $\mathcal N$. To ensure $\mathcal N(\rho_B) = 0$, in the KD setting, the first reference basis is fixed by the eigenbasis of Bob's initial reduced
state $\rho_B$, while the second basis is optimized for each conditional state. In the CV setting, $\mathcal N$ is the Wigner negativity defined with respect to fixed phase-space quadratures. Alice performs a POVM $\{M_a\}$ on her subsystem, and conditioned on the outcome $a$ Bob's state is steered to
\[
\rho_{B|a} = \frac{\Tr_A[(M_a \otimes I)\rho_{AB}]}{p(a)}, \qquad p(a) = \Tr[(M_a \otimes I)\rho_{AB}].
\]
The conditional states satisfy $\rho_B = \sum_a p(a) \rho_{B|a}$ for every choice of POVM. This formula for $\rho_{B|a}$ depends only on Alice's POVM elements $\{M_a\}$, not on which physical instrument realizes them: since Alice's measurement acts on her subsystem alone and no interaction ever touches Bob's system, Bob's conditional state for a given outcome is fixed once its effect $M_a$ is fixed, independently of any further processing of Alice's own post-measurement state. Refining a POVM element, as in Proposition~\ref{prop:refinement} below, is accordingly a statement about splitting an effect $M_{a_0}$ into finer effects $M_{a_0,k}$, not about an arbitrary choice of instrument realizing a fixed set of effects.

Two distinct operational questions can be asked about how much nonclassicality Alice can remotely induce in Bob's system. The first concerns the most nonclassical single conditional state that Alice can prepare, optimized jointly over her measurement and over which outcome is used.

\begin{definition}\label{def:RN}
The maximal steered nonclassicality of $\rho_{AB}$ is
\begin{equation}
\mathcal R_{\mathcal N}(\rho_{AB}) = \max_{\{M_a\},\, a} \mathcal N(\rho_{B|a}),
\end{equation}
where the maximization runs over all POVMs $\{M_a\}$ on Alice's system and over all outcomes $a$ of the chosen POVM.
\end{definition}

The second question concerns the nonclassicality generated on average, per run of the protocol, once the outcome probabilities are taken into account.

\begin{definition}\label{def:CN_general}
The average steered nonclassicality of $\rho_{AB}$ is
\begin{equation}\label{def:CN}
\mathcal C_{\mathcal N}(\rho_{AB}) = \max_{\{M_a\}} \sum_a p(a)\, \mathcal N(\rho_{B|a}),
\end{equation}
where the maximization runs over all POVMs $\{M_a\}$ on Alice's system.
\end{definition}

$\mathcal R_{\mathcal N}$ answers the question of how nonclassical Bob's state can be made to look in the best single branch of the protocol, regardless of how rarely that branch occurs. $\mathcal C_{\mathcal N}$ answers the question of how much nonclassicality is generated per use of the shared state, and is the more operationally relevant quantity whenever the protocol is to be repeated and outcomes cannot be selectively postprocessed, that is, when every outcome $a$ must contribute to the reported average rather than being discarded after the fact. This is compatible with Bob adapting how he accesses the nonclassicality of $\rho_{B|a}$ to the outcome $a$ once it is communicated to him, for instance by evaluating $\mathcal N$ against a second basis chosen depending on $a$, as the constructions of Sec.~\ref{sec:twoqubitKD} do explicitly (Theorem~\ref{thm:CKD}): what "cannot be selectively postprocessed" rules out is discarding unfavorable outcomes from the average, not outcome-dependent choice of how the retained ones are evaluated. Both quantities vanish whenever $\rho_{AB}$ is a product state, since then $\rho_{B|a} = \rho_B$ for every $a$ and every measurement.

\subsection{Ordering and refinement}

\begin{proposition}
\label{prop:ordering}
For every bipartite state $\rho_{AB}$,
\[
0 \leq \mathcal C_{\mathcal N}(\rho_{AB}) \leq \mathcal R_{\mathcal N}(\rho_{AB}).
\]
\end{proposition}

\begin{proof}
Non-negativity of both quantities follows from $\mathcal N \geq 0$. For any fixed POVM $\{M_a\}$,
\begin{align}
\sum_a p(a)\, \mathcal N(\rho_{B|a}) &\leq \sum_a p(a) \max_{a'} \mathcal N(\rho_{B|a'}) \\
&= \max_{a'} \mathcal N(\rho_{B|a'}) \leq \mathcal R_{\mathcal N}(\rho_{AB}).
\end{align}
Since this bound holds for every POVM, it holds in particular for the POVM that maximizes the left-hand side, giving $\mathcal C_{\mathcal N}(\rho_{AB}) \leq \mathcal R_{\mathcal N}(\rho_{AB})$.
\end{proof}

A second structural fact concerns what happens under refinement of Alice's measurement, that is, splitting one POVM element into several. This property was implicit in the choice, made later in this work, to restrict the optimization over Alice's measurements to rank-one elements; we now justify that choice explicitly.

\begin{proposition}
\label{prop:refinement}
Let $\{M_a\}$ be a POVM, and suppose the element $M_{a_0}$ is refined into $K$ positive operators $M_{a_0,1},\ldots,M_{a_0,K}$ with $\sum_{k=1}^K M_{a_0,k} = M_{a_0}$, while all other elements are left unchanged. Then both $\mathcal R_{\mathcal N}$ and $\mathcal C_{\mathcal N}$, evaluated on the refined POVM in place of the original, do not decrease.
\end{proposition}

\begin{proof}
Write $p(a_0,k) = \Tr[(M_{a_0,k}\otimes I)\rho_{AB}]$ and $\rho_{B|a_0,k}$ for the corresponding conditional state, so that $p(a_0) = \sum_k p(a_0,k)$ and
\[
\rho_{B|a_0} = \sum_{k=1}^K \frac{p(a_0,k)}{p(a_0)}\, \rho_{B|a_0,k},
\]
a convex combination of the refined conditional states. By convexity of $\mathcal N$ (Lemma~\ref{lem:convexity}),
\[
\mathcal N(\rho_{B|a_0}) \leq \sum_k \frac{p(a_0,k)}{p(a_0)}\, \mathcal N(\rho_{B|a_0,k}) \leq \max_k \mathcal N(\rho_{B|a_0,k}).
\]
The first inequality, multiplied by $p(a_0)$, shows that the contribution of the outcome $a_0$ to the sum defining $\mathcal C_{\mathcal N}$ cannot decrease under the refinement, since the remaining outcomes are unaffected. The second inequality shows that the maximum defining $\mathcal R_{\mathcal N}$ over all outcomes, including the refined ones, cannot decrease either.
\end{proof}

\begin{corollary}
\label{cor:rankone}
The optimizations defining $\mathcal R_{\mathcal N}(\rho_{AB})$ and $\mathcal C_{\mathcal N}(\rho_{AB})$ can be restricted, without loss of generality, to POVMs whose elements are rank one. Indeed, any POVM element admits a spectral decomposition into positive multiples of rank-one projectors, and by Proposition~\ref{prop:refinement} this refinement can only increase both quantities.
\end{corollary}

A rank-one element of a qubit POVM has the form $E_a={\lambda_a}(\mathbb I+\boldsymbol m_a\cdot\boldsymbol\sigma)/2$ for some direction $\boldsymbol m_a$ and weight $\lambda_a\in(0,1]$, and this corollary only licenses the restriction to such elements; it does not by itself single out the weight-one, two-outcome case $M_{\boldsymbol m}=(\mathbb I+\boldsymbol m\cdot\boldsymbol\sigma)/2$ used for a qubit Alice in Sec.~\ref{sec:twoqubitKD}. The two operational quantities need separate arguments for that further reduction, given in the proofs of Theorem~\ref{thm:sigma_max} and Theorem~\ref{thm:CKD} below.

\subsection{States incapable of remote activation}

We now identify a class of states for which neither quantity can ever be made positive, generalizing the original quantum-classical criterion to a form that applies uniformly to the discrete- and CV settings.

\begin{theorem}
\label{thm:incapable}
Suppose $\rho_{AB}$ admits a decomposition
\[
\rho_{AB} = \sum_i p_i\, \rho_i^A \otimes \sigma_i^B,
\]
where $\{p_i\}$ is a probability distribution, each $\rho_i^A$ is a state on Alice's system, and each $\sigma_i^B$ satisfies $\mathcal N(\sigma_i^B) = 0$. Then
\[
\mathcal R_{\mathcal N}(\rho_{AB}) = \mathcal C_{\mathcal N}(\rho_{AB}) = 0.
\]
\end{theorem}

\begin{proof}
For an arbitrary POVM $\{M_a\}$ on Alice's system, the conditional state on Bob's side is
\[
\rho_{B|a} = \sum_i \frac{p_i \Tr(M_a \rho_i^A)}{p(a)}\, \sigma_i^B,
\]
which is a convex combination of the operators $\sigma_i^B$, since the coefficients are non-negative and sum to one. By Lemma~\ref{lem:convexity} the classical set is convex, so $\rho_{B|a}$ remains in the classical set for every outcome $a$ and every choice of $\{M_a\}$, that is $\mathcal N(\rho_{B|a}) = 0$ identically. This immediately gives $\mathcal R_{\mathcal N}(\rho_{AB}) = 0$, and Proposition~\ref{prop:ordering} then gives $\mathcal C_{\mathcal N}(\rho_{AB}) = 0$ as well.
\end{proof}

Two special cases of Theorem~\ref{thm:incapable} are worth stating separately, since they correspond directly to the two settings treated later in this work.

\begin{corollary}[DV case]
\label{cor:dv}
Let $\rho_{AB} = \sum_i p_i\, \rho_i^A \otimes |a_i\rangle\langle a_i|_B$ be a quantum-classical state with respect to an orthonormal basis $\{|a_i\rangle\}$ of Bob's system, and let the KD nonclassicality be evaluated with $\{|a_i\rangle\}$ as the first reference basis. Then $\rho_{B|a}$ remains diagonal in $\{|a_i\rangle\}$ for every measurement outcome, so that $\sigma_i^B = |a_i\rangle\langle a_i|$ trivially satisfies $\mathcal N_{KD}(\sigma_i^B) = 0$, and $\mathcal R_{KD}(\rho_{AB}) = \mathcal C_{KD}(\rho_{AB}) = 0$.
\end{corollary}

\begin{corollary}[CV case]
\label{cor:cv}
Let $\rho_{AB} = \sum_i p_i\, \rho_i^A \otimes \sigma_i^B$, where each $\sigma_i^B$ is individually Wigner-positive, for instance a coherent, squeezed, or thermal state. Then $\mathcal R_W(\rho_{AB}) = \mathcal C_W(\rho_{AB}) = 0$.
\end{corollary}

The two corollaries rest on different mechanisms. In the discrete-variable case, orthogonality of the branch states $\{|a_i\rangle\}$ automatically renders each branch classical with respect to a first reference basis chosen to coincide with that same set of states, so any orthogonal quantum-classical decomposition suffices. In the CV case there is no such automatic guarantee: an orthogonal basis of Bob's Hilbert space need not consist of individually Wigner-positive states. The Fock basis is the standard counterexample, since $|n\rangle\langle n|$ has a negative Wigner function for every $n \geq 1$; a state of the form $\rho_{AB} = \sum_n p_n \rho_n^A \otimes |n\rangle\langle n|_B$ is diagonal on Bob's side in an orthonormal basis, yet Bob's own reduced state already carries Wigner negativity as soon as any $p_{n\geq 1} > 0$, let alone the conditional states. Theorem~\ref{thm:incapable} therefore isolates the correct sufficient condition, individual classicality of the branch states $\sigma_i^B$, of which orthogonality is neither necessary nor sufficient.

\begin{remark} \label{remark:1}
Theorem~\ref{thm:incapable} gives a sufficient condition for both quantities to vanish. Whether the condition is also necessary, that is, whether $\mathcal R_{\mathcal N}(\rho_{AB}) = 0$ forces $\rho_{AB}$ into a decomposition of the stated form, is a more delicate question. It is not simply the statement that $\rho_{AB}$ is separable, since separability alone permits the branch states $\sigma_i^B$ to be individually nonclassical while still combining, for a specific measurement, into a classical conditional state; the requirement here is that no POVM element on Alice's side, however chosen, extracts a nonclassical branch. This is a strictly finer requirement than existence of a local-hidden-state model in the usual Einstein-Podolsky-Rosen sense, since the latter permits the hidden states themselves to be arbitrary quantum states rather than classical ones. For two-qubit systems the question can be resolved geometrically, since the full set of states reachable by Alice's measurements is characterized by the steering ellipsoid; we return to this point in Sec.~\ref{sec:twoqubitKD}.
\end{remark}

\subsection{Optimal shared states}

We next identify the states for which the two quantities attain their largest possible values. Both results rely on the HJW theorem~\cite{hughston1993complete}, which guarantees that for a purification of full Schmidt rank, every ensemble decomposition of the reduced state on one side can be remotely prepared by a suitable POVM on the other side.

\begin{theorem}
\label{thm:optimal}
Let $\mathcal H_B$ be finite dimensional, $d=\dim\mathcal H_B<\infty$, and let $|\Psi\rangle_{AB} = \sum_{i=1}^d \lambda_i |\phi_i\rangle_A |a_i\rangle_B$ be a pure state of full Schmidt rank, $\lambda_i > 0$ for every $i$ (so in particular $\dim\mathcal H_A\ge d$, as required for $d$ orthonormal Schmidt vectors $\{|\phi_i\rangle_A\}$ to exist on Alice's side). Then:
\begin{enumerate}
\item[(i)] Alice can remotely prepare any pure state supported on $\mathcal H_B$, and consequently
\[
\mathcal R_{\mathcal N}(\Psi) = \max_{|\psi\rangle \in \mathcal H_B} \mathcal N(|\psi\rangle\langle\psi|).
\]
\item[(ii)] Alice can remotely prepare any ensemble decomposition of $\rho_B = \Tr_A|\Psi\rangle\langle\Psi|$ into pure states of $\mathcal H_B$, and consequently
\[
\mathcal C_{\mathcal N}(\Psi) = \max_{\{q_j,\,|\psi_j\rangle\} :\ \sum_j q_j |\psi_j\rangle\langle\psi_j| = \rho_B} \ \sum_j q_j\, \mathcal N(|\psi_j\rangle\langle\psi_j|).
\]
\end{enumerate}
\end{theorem}

\begin{proof}
By the HJW theorem, for a purification of full Schmidt rank, every ensemble decomposition $\rho_B = \sum_j q_j |\psi_j\rangle\langle\psi_j|$ with $|\psi_j\rangle \in \mathcal H_B$ can be remotely prepared: there exists a POVM $\{M_j\}$ on Alice's system such that the outcome $j$ occurs with probability $q_j$ and steers Bob to $\rho_{B|j} = |\psi_j\rangle\langle\psi_j|$.

For part (i), since $\mathcal R_{\mathcal N}(\Psi)$ is obtained by optimizing over all POVMs and all outcomes, Alice may in particular choose a decomposition that contains the maximally nonclassical pure state $|\psi^\ast\rangle = \arg\max_{|\psi\rangle \in \mathcal H_B} \mathcal N(|\psi\rangle\langle\psi|)$ as one of its branches, occurring with any sufficiently small probability $p>0$. This uses finite-dimensionality essentially: because $d<\infty$ and $\lambda_i>0$ for every $i$, $\rho_B=\Tr_A|\Psi\rangle\langle\Psi|$ is a strictly positive $d\times d$ matrix, hence invertible with bounded inverse $\rho_B^{-1}$, and for any unit vector $|\psi^\ast\rangle\in\mathcal H_B$ the standard Loewner-order criterion $\rho_B-p\,|\psi^\ast\rangle\langle\psi^\ast|\succeq0 \iff p\le1/\langle\psi^\ast|\rho_B^{-1}|\psi^\ast\rangle$ gives a strictly positive threshold, since $\langle\psi^\ast|\rho_B^{-1}|\psi^\ast\rangle\le\|\rho_B^{-1}\|_{\mathrm{op}}<\infty$ is finite for every unit vector by boundedness of $\rho_B^{-1}$. So $\rho_B=p\,|\psi^\ast\rangle\langle\psi^\ast|+(1-p)\tau$ with $\tau=(\rho_B-p\,|\psi^\ast\rangle\langle\psi^\ast|)/(1-p)$ a valid density operator for all sufficiently small $p>0$, and this decomposition is completed by any further decomposition of $\tau$ into pure states of the residual weight. This achieves $\mathcal N(\rho_{B|a^\ast}) = \mathcal N(|\psi^\ast\rangle\langle\psi^\ast|)$ for the corresponding outcome $a^\ast$, and no POVM can do better since every conditional state $\rho_{B|a}$ is supported on $\mathcal H_B$, giving the stated equality. The corresponding outcome probability may be arbitrarily small, which is consistent with
the postselected definition of $\mathcal R_{\mathcal N}$.

For part (ii), since $\mathcal C_{\mathcal N}(\Psi)$ optimizes the probability-weighted average of $\mathcal N$ over all POVMs, and every decomposition of $\rho_B$ into pure states on $\mathcal H_B$ is achievable as the outcome ensemble of some POVM by the HJW theorem, the optimization over POVMs is equivalent to the optimization over pure-state decompositions of $\rho_B$ stated on the right-hand side.
\end{proof}

Part (ii) identifies $\mathcal C_{\mathcal N}(\Psi)$, for a full-Schmidt-rank purification, with the maximal pure-state ensemble extension of $\mathcal N$, analogous to a concave-roof or assisted-resource construction, evaluated on $\rho_B$. This is the nonclassicality analogue of entanglement of assistance and of assisted coherence, quantities defined in exactly the same way, as the largest average resource achievable over all ensemble decompositions of a fixed reduced state, in contrast to measures such as entanglement of formation which instead minimize over decompositions. 
We refer to $\mathcal C_{\mathcal N}^{\mathrm a}(\rho_B) \equiv \max_{\{q_j,\psi_j\}} \sum_j q_j \mathcal N(|\psi_j\rangle\langle\psi_j|)$ as the nonclassicality of assistance of $\rho_B$, using the superscript $\mathrm a$ to keep this single-system quantity visually distinct from the bipartite $\mathcal C_{\mathcal N}(\rho_{AB})$ of Definition~\ref{def:CN_general}, even though Theorem~\ref{thm:optimal}(ii) shows the two agree for a full-Schmidt-rank purification.

\begin{remark} \label{remark:2}
Theorem~\ref{thm:optimal} is stated and proved only for finite-dimensional $\mathcal H_B$; this is essential to the proof of part (i), which relies on $\rho_B^{-1}$ being a bounded operator, not merely a simplifying restriction. For a genuinely infinite-dimensional $\mathcal H_B$, a full-rank $\rho_B$ (trivial kernel, every eigenvalue $\lambda_i>0$) need not have a bounded inverse: since $\rho_B$ is trace class, its eigenvalues necessarily accumulate at $0$, so $\rho_B^{-1}$, defined by spectral calculus, is an unbounded operator whose domain is a dense but proper subspace of $\mathcal H_B$. For a target state $|\psi^\ast\rangle$ outside this domain, $\langle\psi^\ast|\rho_B^{-1}|\psi^\ast\rangle=\sum_i|\langle a_i|\psi^\ast\rangle|^2/\lambda_i=\infty$, and no $p>0$ satisfies $\rho_B-p\,|\psi^\ast\rangle\langle\psi^\ast|\succeq0$: such a $|\psi^\ast\rangle$ cannot enter any ensemble decomposition of $\rho_B$ with nonzero weight, however small. The proof of Theorem~\ref{thm:optimal}(i) therefore does not extend to a general infinite-dimensional $\mathcal H_B$ without a further, explicit assumption relating $|\psi^\ast\rangle$ to the spectrum of $\rho_B$, and we claim no such extension.

This is precisely why Sec.~\ref{sec:steered_wigner} never invokes Theorem~\ref{thm:optimal} on the full, infinite-dimensional oscillator Hilbert space: the marginal $\rho_B$ constructed there is exactly rank two by design, Eq.~\eqref{eq:bob_marginal_cat}, supported on the finite-dimensional subspace $\mathrm{span}\{\ket\alpha,\ket{-\alpha}\}$. Theorem~\ref{thm:optimal} is applied with $\mathcal H_B$ taken to be this two-dimensional subspace itself, so its finite-dimensional hypothesis is satisfied exactly, with no extension required, the Hilbert-space restriction already flagged where the theorem is used, Sec.~\ref{sec:steered_wigner}A. Within that restricted setting, $\max_{|\psi\rangle\in\mathrm{span}\{\ket\alpha,\ket{-\alpha}\}}\mathcal N_W(|\psi\rangle\langle\psi|)$ is finite and attained, since the maximization is again over a compact set, exactly as in the DV case. The unrestricted supremum $\max_{|\psi\rangle\in\mathcal H_B}\mathcal N_W(|\psi\rangle\langle\psi|)$ over the entire oscillator Hilbert space, by contrast, is unbounded: Wigner negativity can be made arbitrarily large by states with sufficiently fine phase-space structure, the negativity volume of Fock states $|n\rangle$ growing without bound as $n\to\infty$~\cite{kenfack2004negativity}. This unrestricted quantity is never computed or claimed here; only the finite-dimensional, rank-restricted construction of Sec.~\ref{sec:steered_wigner} is used.
\end{remark}

\subsection{Invariance under local unitary operations}

Finally, we examine the behaviour of both quantities under local unitary transformations $\rho_{AB} \mapsto (U_A \otimes U_B)\rho_{AB}(U_A^\dagger \otimes U_B^\dagger)$.

\begin{theorem}
\label{thm:lu}
For any local unitary $U_A$ on Alice's system,
\begin{align}
\mathcal R_{\mathcal N}\big((U_A\otimes I)\rho_{AB}(U_A^\dagger\otimes I)\big) = \mathcal R_{\mathcal N}(\rho_{AB}), \\
\mathcal C_{\mathcal N}\big((U_A\otimes I)\rho_{AB}(U_A^\dagger\otimes I)\big) = \mathcal C_{\mathcal N}(\rho_{AB}),
\end{align}
for both quantities and either choice of $\mathcal N$. For a local unitary $U_B$ on Bob's system, the same equalities hold in the KD case under the convention used throughout this work, in which the first KD reference basis is defined by the eigenbasis of Bob's reduced state and therefore co-rotates under
$U_B$, and in the Wigner case whenever $U_B$ is a Gaussian unitary.
\end{theorem}

\begin{proof}
Invariance under $U_A$ is immediate: relabelling Alice's POVM as $M_a' = U_A M_a U_A^\dagger$ gives the same set of conditional states $\rho_{B|a}$ and the same outcome probabilities as the original POVM $\{M_a\}$ applied to $\rho_{AB}$, so the sets of values over which both optimizations run are unchanged.

For $U_B$, a POVM $\{M_a\}$ applied to $(U_A\otimes U_B)\rho_{AB}(U_A^\dagger\otimes U_B^\dagger)$ steers Bob to $U_B \rho_{B|a} U_B^\dagger$ with the same outcome probability as $\{M_a\}$ applied to $\rho_{AB}$ itself. The two quantities are therefore invariant under $U_B$ if and only if $\mathcal N(U_B \sigma U_B^\dagger) = \mathcal N(\sigma)$ for every state $\sigma$ on Bob's system.

In the KD case considered throughout this work, the first reference basis is fixed by the eigenbasis of Bob's initial reduced state $\rho_B$, while the second basis is optimized freely. Concretely, for the general state $\sigma$ of the criterion in the previous paragraph (Bob's reduced state $\rho_B$ itself, when evaluating $\mathcal N_{KD}$ at the marginal level, or a steered conditional state $\rho_{B|a}$, when evaluating it branch by branch), the first basis is by definition the eigenbasis of $\rho_B$ specifically, not of $\sigma$; under $\sigma \mapsto U_B \sigma U_B^\dagger$ this reference eigenbasis rotates to the eigenbasis of $U_B\rho_B U_B^\dagger$, i.e. is simply rotated by $U_B$, exactly tracking the state it is attached to. The optimal second basis, being unconstrained, rotates correspondingly, so the value of the optimized functional is left unchanged for arbitrary $U_B$.

In the Wigner case, $\mathcal N_W$ is defined with respect to a fixed pair of quadratures $(x,p)$ that does not co-rotate with $U_B$. If $U_B$ is a Gaussian unitary, generated by a Hamiltonian at most quadratic in the mode operators, its action on phase space is an affine symplectic transformation, which is measure preserving and maps the Wigner function of $\sigma$ to the pulled-back function evaluated at the transformed phase-space point; the integral $\int |W_\sigma| \, dx\,dp$ is therefore unchanged. This argument does not extend to a general, non-Gaussian $U_B$: such a unitary can map a Wigner-positive state to a Wigner-negative one, as is the case for any unitary that prepares a Schr\"odinger-cat-like superposition from a coherent state, so $\mathcal N_W(U_B\sigma U_B^\dagger) \neq \mathcal N_W(\sigma)$ in general.
\end{proof}

Theorem~\ref{thm:lu} shows that both quantities are intrinsic to the correlations in $\rho_{AB}$, independent of Alice's local reference frame in either setting, and independent of Bob's local reference frame whenever $\mathcal N$ is itself defined relative to that frame in a co-rotating way, as is the case for the KD construction used here. In the Wigner case, by contrast, the choice of quadratures on Bob's side is physical rather than a free labelling convention, and only the Gaussian subgroup of local operations on Bob leaves the figure of merit unchanged; this restriction should be kept in mind whenever the CV results of Sec.~\ref{sec:steered_wigner} are compared across different local Gaussian frames.

\section{Remote activation of KD nonclassicality for two-qubit states}
\label{sec:twoqubitKD}

We now specialize the general framework to two-qubit systems and consider KD nonclassicality as the resource generated on Bob's side. The goal of this section is to obtain closed expressions for both the maximal and the average KD nonclassicality that Alice can remotely induce in Bob's conditional qubit state.

\subsection{General setup}

Any two-qubit state can be written as
\begin{align}
\rho_{AB} = \frac14\Big[\mathbb I\otimes\mathbb I + \boldsymbol a\cdot\boldsymbol\sigma\otimes\mathbb I + \mathbb I\otimes\boldsymbol b\cdot\boldsymbol\sigma + \sum_{i,j=1}^3 T_{ij}\,\sigma_i\otimes\sigma_j\Big],
\label{eq:bloch_twoqubit}
\end{align}
where $\boldsymbol a$ and $\boldsymbol b$ are the local Bloch vectors of Alice and Bob and $T$ is the correlation matrix. Bob's reduced state is $\rho_B = \tfrac12(\mathbb I + \boldsymbol b\cdot\boldsymbol\sigma)$. We assume $\boldsymbol b \neq 0$, so that the eigenbasis of $\rho_B$ is uniquely defined by the direction $\boldsymbol n_B = \boldsymbol b/|\boldsymbol b|$; the degenerate case $\boldsymbol b = 0$, for which $\rho_B$ is maximally mixed and its eigenbasis is not unique, is discussed separately at the end of this section. Because $\boldsymbol n_B$ is tied to Bob's own marginal in this way, the KD nonclassicality studied below is a \emph{marginal-adapted} quantity: it is evaluated relative to a reference basis that is itself a function of the state under consideration, and that consequently moves when the marginal does, for instance under the local channels studied in Sec.~\ref{sec:local_ops}. We flag this explicitly here because it is what makes the comparisons of Sec.~\ref{sec:local_ops} well posed: what is compared there is this marginal-adapted quantity evaluated before and after a channel, each time relative to the reference direction appropriate to that state, not a single fixed-basis functional evaluated on two different states.

By Corollary~\ref{cor:rankone}, both optimizations may be restricted without loss of generality to rank-one POVM elements $E_{\boldsymbol m}=\lambda\,(\mathbb I+\boldsymbol m\cdot\boldsymbol\sigma)/2$, $|\boldsymbol m|=1$, $\lambda\in(0,1]$. Conditioned on such an outcome, Bob is steered to $\rho_{B|\boldsymbol m} = \tfrac12(\mathbb I+\boldsymbol r(\boldsymbol m)\cdot\boldsymbol\sigma)$ with
\begin{align}
\boldsymbol r(\boldsymbol m) = \frac{\boldsymbol b + T^T\boldsymbol m}{1+\boldsymbol a\cdot\boldsymbol m}, \qquad p(\boldsymbol m) = \frac{1+\boldsymbol a\cdot\boldsymbol m}{2},
\label{eq:steered_bloch}
\end{align}
defined for every outcome of nonzero probability, $p(\boldsymbol m)>0$ (the denominator is shown below, in the proof of Corollary~\ref{cor:vanishing_iff}, to vanish for at most a single point $\boldsymbol m\in S^2$),
independently of $\lambda$, so every direction $\boldsymbol m\in S^2$ is reachable this way whatever weight and completing outcomes are chosen; as $\boldsymbol m$ ranges over the Bloch sphere, $\boldsymbol r(\boldsymbol m)$ traces out Bob's steering ellipsoid $\mathcal E$. This alone reduces $\mathcal R_{KD}$ to a single-branch maximization over $\mathcal E$, since Definition~\ref{def:RN} never involves an outcome probability; the further reduction to the weight-one dichotomic pair $\{M_{\boldsymbol m},M_{-\boldsymbol m}\}$, $M_{\boldsymbol m}=\tfrac12(\mathbb I+\boldsymbol m\cdot\boldsymbol\sigma)$, used for both quantities below, is justified separately in the proofs of Theorem~\ref{thm:RKD} and Theorem~\ref{thm:CKD}, exactly as already noted after Corollary~\ref{cor:rankone}.

The reference KD basis is fixed at $\boldsymbol n_B$, the eigenbasis of $\rho_B$. For a steered Bloch vector $\boldsymbol r$, write
\begin{align}
r_\parallel = \boldsymbol r\cdot\boldsymbol n_B, \qquad r_\perp = |\boldsymbol r\times\boldsymbol n_B|,
\end{align}
so that in the reference basis,
\begin{align}
\rho(\boldsymbol r) = \frac12\begin{pmatrix} 1+r_\parallel & r_\perp e^{-i\varphi} \\ r_\perp e^{i\varphi} & 1-r_\parallel \end{pmatrix},
\end{align}
with $\varphi$ the azimuthal phase of $\boldsymbol r$ about $\boldsymbol n_B$.

\subsection{KD distribution of the steered qubit state}

The second KD basis is optimized over all orthonormal qubit bases. A notational point worth fixing here: for a fixed pair of bases, $\mathcal N_{KD}$ as defined by Eq.~\eqref{eq:N_measure} is a function of $\rho$ and of both bases; from this point on, however, $\mathcal N_{KD}(\rho)$ denotes that quantity already maximized over the second basis at fixed reference basis $\boldsymbol n_B$, i.e.\ $\mathcal N_{KD}(\rho)\equiv\max_{\{b_k\}}\mathcal N_{KD}(\rho;\boldsymbol n_B,\{b_k\})$, consistently with its use in Proposition~\ref{prop:no_anomaly_before} and throughout. The next lemma gives $|Q_{ik}|^2$ explicitly in terms of that basis's polar and azimuthal angles and the steered Bloch vector's components.
\begin{lemma}
\label{lem:qik}
Let $\boldsymbol r$ be the steered Bloch vector, $r_\parallel=\boldsymbol r\cdot\boldsymbol n_B$, $r_\perp=|\boldsymbol r\times\boldsymbol n_B|$, and $\varphi$ the azimuthal phase of $\boldsymbol r$ about $\boldsymbol n_B$. Work in the frame with $\boldsymbol n_B=\hat z$ and the transverse part of $\boldsymbol r$ along $\hat x$, and let the second measurement basis be the eigenbasis of $\hat u(\theta,\phi)\cdot\boldsymbol\sigma$, with $\hat u(\theta,\phi)=(\sin\theta\cos\phi,\sin\theta\sin\phi,\cos\theta)$ an arbitrary Bloch direction. Write $\psi\equiv\phi-\varphi$ for the azimuth of this second basis relative to the coherence phase of $\boldsymbol r$, $\Pi_i^{n_B}=\frac12(\mathbb I+(-1)^i\boldsymbol n_B\cdot\boldsymbol\sigma)$, $\Pi_k^{\hat u}=\frac12(\mathbb I+(-1)^k\hat u\cdot\boldsymbol\sigma)$, and $Q_{ik}=\Tr[\Pi_k^{\hat u}\Pi_i^{n_B}\rho]$ for $\rho=\frac12(\mathbb I+\boldsymbol r\cdot\boldsymbol\sigma)$. Then
\begin{align}
16\,|Q_{ik}|^2 = P_{ik}^2 + a^2 + 2\epsilon_{ik}P_{ik}\,a\cos\psi,
\end{align}
where $a = r_\perp\sin\theta$, and
\begin{align}
\begin{aligned}
P_{00} &= (1+\cos\theta)(1+r_\parallel), & \epsilon_{00} &= +1,\\
P_{11} &= (1+\cos\theta)(1-r_\parallel), & \epsilon_{11} &= -1,\\
P_{01} &= (1-\cos\theta)(1+r_\parallel), & \epsilon_{01} &= -1,\\
P_{10} &= (1-\cos\theta)(1-r_\parallel), & \epsilon_{10} &= +1.
\end{aligned}
\end{align}
\end{lemma}
\begin{proof}
Using $Q_{ik} = \Tr[\Pi_k^{\hat u}\Pi_i^{n_B}\rho]$ together with the Pauli identity $(\hat u\cdot\boldsymbol\sigma)(\boldsymbol n_B\cdot\boldsymbol\sigma) = \hat u\cdot\boldsymbol n_B\,\mathbb I + i(\hat u\times\boldsymbol n_B)\cdot\boldsymbol\sigma$ gives, after expanding and taking the trace against $\rho = \tfrac12(\mathbb I+\boldsymbol r\cdot\boldsymbol\sigma)$,
\begin{align*}
Q_{ik} &= \frac14\Big[1+(-1)^{i+k}\cos\theta+(-1)^i r_\parallel \\
&\quad+(-1)^k\big(r_\perp\sin\theta\cos\psi+r_\parallel\cos\theta\big) \\
& \quad- i(-1)^{i+k}r_\perp\sin\theta\sin\psi\Big].
\end{align*}
Squaring the real and imaginary parts and simplifying, using $\cos^2\psi+\sin^2\psi=1$, gives $16|Q_{ik}|^2 = P_{ik}^2+a^2+2\epsilon_{ik}P_{ik}a\cos\psi$ with $P_{ik},\epsilon_{ik}$ as stated.
\end{proof}

\subsection{Optimal second basis}
 
We now identify, in closed form and for every physical steered state, the
second measurement basis that maximizes
$\Sigma(\theta,\psi) \equiv \sum_{i,k=0}^1 |Q_{ik}|$, and with it the KD
sum-negativity $\mathcal N_{KD}$, over all choices of that basis. As before
we write $c = \cos\theta$, $s = \sin\theta = \sqrt{1-c^2} \ge 0$ for
$\theta \in [0,\pi]$, and $x = \cos\psi$. By Lemma~\ref{lem:qik}, $\Sigma$
depends on $\psi$ only through $x$, so we treat $\Sigma$ as a function of
the two independent real variables $\theta$ and $x$. For brevity we also
write
\begin{equation}
  A = 1+r_\parallel, \qquad
  B = 1-r_\parallel, \qquad
  q = r_\perp ,
\end{equation} 
so that $A,B\ge0$ and $A+B=2$ for every physical state (since $r_\parallel^2+r_\perp^2\le1$), with $A=0$ or $B=0$ occurring only at the poles $r_\parallel=\mp1$, treated separately in the proof of Theorem~\ref{thm:sigma_max} below.
With this notation, and writing
$a = q s$, Lemma~\ref{lem:qik} states
\begin{equation}
  16 \, |Q_{ik}|^2 = P_{ik}^2 + a^2 + 2\epsilon_{ik} P_{ik} \, a \, x ,
  \label{eq:Qik_recall}
\end{equation}
with $P_{00} = A(1+c)$, $P_{01} = A(1-c)$, $P_{10} = B(1-c)$,
$P_{11} = B(1+c)$, and $\epsilon_{00} = \epsilon_{10} = +1$,
$\epsilon_{01} = \epsilon_{11} = -1$. It is convenient to name the four
square roots that appear in $\Sigma$. Define
\begin{align}
  X_\pm &= \sqrt{A^2(1\pm c)^2 + a^2 \pm 2A(1\pm c) \, a \, x} , \\
  Y_\pm &= \sqrt{B^2(1\mp c)^2 + a^2 \pm 2B(1\mp c) \, a \, x} ,
\end{align}
so that $X_\pm = 4|Q_{0,0/1}|$ and $Y_\pm = 4|Q_{1,1/0}|$, and
\begin{equation}
  \Sigma(\theta,\psi) = \tfrac14 \big( X_+ + X_- + Y_+ + Y_- \big) .
  \label{eq:Sigma_split}
\end{equation}
The sign pattern in these definitions simply records
Eq.~\eqref{eq:Qik_recall}: $X_+$ is built from $P_{00}$, which carries the
weight $1+c$ and the sign $\epsilon_{00}=+1$, while $X_-$ is built from
$P_{01}$, with weight $1-c$ and sign $\epsilon_{01}=-1$. For the second
pair the roles of the two weights are exchanged, $Y_+$ carries weight
$1-c$ (from $P_{10}$) and $Y_-$ carries weight $1+c$ (from $P_{11}$),
again exactly as in Eq.~\eqref{eq:Qik_recall}. Recall also that KD
sum-negativity is defined through
$\mathcal N_{KD}(\rho) = \tfrac12 \big( \sum_{i,k} |Q_{ik}| - 1 \big)$, so
maximizing $\mathcal N_{KD}$ over the second basis is equivalent to
maximizing $\Sigma$.

Dividing $X_+^2$ by $1+c$ and $X_-^2$ by $1-c$ cancels the cross terms in $x$, since they enter $X_+^2$ and $X_-^2$ with opposite sign and the same weight.
We isolate this fact as a lemma, since it is the only nontrivial
ingredient of the proof and it applies unchanged to both pairs of terms in
Eq.~\eqref{eq:Sigma_split}.
 
\begin{lemma}[Weighted Cauchy--Schwarz bound]
\label{lem:cs_bound}
Let $c \in (-1,1)$, $q \ge 0$, $A>0$, and let $x$ be an arbitrary real
number. With $s = \sqrt{1-c^2}$ and $a = qs$, define
\begin{align}
  X_+ = \sqrt{A^2(1+c)^2+a^2+2A(1+c)ax}, \\
  X_- = \sqrt{A^2(1-c)^2+a^2-2A(1-c)ax} .
\end{align}
Then, for every such $x$,
\begin{equation}
  X_+ + X_- \le 2\sqrt{A^2+q^2} ,
  \label{eq:cs_bound}
\end{equation}
with equality if and only if
\begin{equation}
  A s x = q c .
  \label{eq:cs_equality}
\end{equation}
\end{lemma}
 
\begin{proof}
The proof is given in Appendix~\ref{app:cs_bound}.
\end{proof}
 
Lemma~\ref{lem:cs_bound} applies directly to the pair $(X_+,X_-)$, using
the weights $1+c$ and $1-c$ in that order. It also applies to the pair
$(Y_+,Y_-)$, after replacing $A$ by $B$ and $c$ by $-c$: since $s$ and
$a=qs$ are unchanged under $c \to -c$, this substitution turns
$X_\pm(A,c)$ into $Y_\pm(B,-c) = Y_\pm$ term by term. We therefore obtain
two bounds, holding simultaneously for every $\theta\in(0,\pi)$ and every
$\psi$,
\begin{equation}
  X_++X_- \le 2\sqrt{A^2+q^2} , \qquad
  Y_++Y_- \le 2\sqrt{B^2+q^2} .
  \label{eq:two_bounds}
\end{equation}
 
\begin{theorem}[Maximal KD sum-negativity of a steered qubit]
\label{thm:sigma_max}
For every physical steered state $\rho(\boldsymbol r)$, that is,
$r_\parallel^2+r_\perp^2 \le 1$, and every choice of the second measurement
basis $(\theta,\psi) \in [0,\pi] \times [0,2\pi)$,
\begin{align}
  \mathcal N_{KD}\big(\rho(\boldsymbol r)\big)
  \le \frac14 \Big[ &\sqrt{(1+r_\parallel)^2+r_\perp^2} 
  + \sqrt{(1-r_\parallel)^2+r_\perp^2} \Big] - \frac12 .
  \label{eq:qubit_KD}
\end{align}
Equality is attained at $\theta=\pi/2$ and $\psi \in \{\pi/2, 3\pi/2\}$,
that is, when the second basis is mutually unbiased with respect to
$\boldsymbol n_B$ and its Bloch direction is perpendicular to
$\boldsymbol n_B$. If $r_\perp>0$, this is the only maximizer, up to the
reflection $\psi \mapsto \psi+\pi$.
\end{theorem}
 
\begin{proof}
The proof is somewhat long and is given in Appendix~\ref{app:sigma_max}.
\end{proof}

Theorem~\ref{thm:sigma_max}'s bound is proved for every $\theta\in[0,\pi]$ and $\psi\in[0,2\pi)$ at once; locating the optimum requires no search over this domain.
When $r_\perp=0$, the steered state is
diagonal in the reference basis for every choice of $\boldsymbol n_B$,
and Eq.~\eqref{eq:qubit_KD} gives $\mathcal N_{KD}=\tfrac14(A+B)-\tfrac12=0$,
so such a state carries no KD nonclassicality for any second basis, once
optimized, as expected. When $r_\parallel=0$ and $r_\perp=1$, the steered
state is a pure state on the equator of the Bloch sphere, and the
reference and optimal second basis form the canonical pair of mutually
unbiased qubit bases; Eq.~\eqref{eq:qubit_KD} then gives
\begin{equation}
  \mathcal N_{KD} = \frac{\sqrt2-1}{2} \approx 0.207 ,
\end{equation}
which is the known maximal KD sum-negativity for a pair of mutually
unbiased qubit bases.

\subsection{Maximal and average steered KD nonclassicality}

Combining Theorem~\ref{thm:sigma_max} with the steering map \eqref{eq:steered_bloch} gives closed expressions for both operational quantities of Sec.~\ref{sec:properties}. Throughout, ``closed form" is used in the sense that the abstract optimization over Alice's POVM in Definitions~\ref{def:RN}-\ref{def:CN_general} has been replaced by an explicit formula involving only elementary functions of the state parameters and a residual optimization over the two-parameter steering ellipsoid $\mathcal E$ (equivalently, Alice's Bloch-sphere measurement direction $\boldsymbol m$); it is this residual, low-dimensional, explicitly parametrized optimization, evaluated numerically in Secs.~\ref{sec:examples}-\ref{sec:local_ops}, that is meant whenever a maximization over $\boldsymbol m$ or over $c=\cos\theta$ still appears below.

\begin{theorem}
\label{thm:RKD}
The maximal steered KD nonclassicality of a two-qubit state $\rho_{AB}$ is
\begin{align}
\mathcal R_{KD}(\rho_{AB}) &= \max_{|\boldsymbol m|=1}\ \max_{s=\pm1}\ \frac14\Big[\sqrt{(1+s\,r_\parallel(s\boldsymbol m))^2+r_\perp^2(s\boldsymbol m)} \nonumber\\
&\quad +\sqrt{(1-s\,r_\parallel(s\boldsymbol m))^2+r_\perp^2(s\boldsymbol m)}\Big] - \frac12,
\label{eq:RKD}
\end{align}
where $r_\parallel(\boldsymbol m) = \boldsymbol r(\boldsymbol m)\cdot\boldsymbol n_B$, $r_\perp(\boldsymbol m) = |\boldsymbol r(\boldsymbol m)\times\boldsymbol n_B|$, and $\boldsymbol r(\boldsymbol m)$ is given by Eq.~\eqref{eq:steered_bloch}. 
\end{theorem}

\begin{proof} 
By Corollary~\ref{cor:rankone}, the optimization over Alice's POVM in Definition~\ref{def:RN} may be restricted to rank-one elements $E_a=\lambda_a(\mathbb I+\boldsymbol m_a\cdot\boldsymbol\sigma)/2$. For such an element the conditional state on Bob's side is
\begin{align*}
\rho_{B|a}
&=\frac{\mathrm{Tr}_A[(E_a\otimes\mathbb I)\rho_{AB}]}{\mathrm{Tr}[(E_a\otimes\mathbb I)\rho_{AB}]}\\
&=\frac{\mathrm{Tr}_A[(\Pi_{\boldsymbol m_a}\otimes\mathbb I)\rho_{AB}]}{\mathrm{Tr}[(\Pi_{\boldsymbol m_a}\otimes\mathbb I)\rho_{AB}]}\\
&=\rho(\boldsymbol r(\boldsymbol m_a)) ,
\end{align*}
independent of $\lambda_a$, since the weight cancels between numerator and denominator; here $\Pi_{\boldsymbol m}=\frac12(\mathbb I+\boldsymbol m\cdot\boldsymbol\sigma)$. So a rank-one branch along direction $\boldsymbol m$ always produces the same conditional state, regardless of what weight it carries or how many other outcomes complete the POVM, and every direction $\boldsymbol m\in S^2$ is realizable this way, for instance by completing $\{\lambda\Pi_{\boldsymbol m},\mathbb I-\lambda\Pi_{\boldsymbol m}\}$ at any $\lambda\in(0,1]$. Since Definition~\ref{def:RN} maximizes over a single outcome and does not involve the outcome probability, this shows $\mathcal R_{KD}=\max_{\boldsymbol m\in S^2}\mathcal N_{KD}(\rho(\boldsymbol r(\boldsymbol m)))$, with the maximum in particular attained by the dichotomic pair $\{M_{\boldsymbol m},M_{-\boldsymbol m}\}$ at the optimal $\boldsymbol m$, and the outer maximization below may be taken over $s=\pm1$ together with $\boldsymbol m$ ranging over a hemisphere, without loss of generality. Composing this with Theorem~\ref{thm:sigma_max}, already maximized over the second KD basis, gives Eq.~\eqref{eq:RKD}. Since $\boldsymbol m$ ranges over the full Bloch sphere, $\boldsymbol r(\boldsymbol m)$ for $\boldsymbol m$ and $\boldsymbol r(-\boldsymbol m)$ for $-\boldsymbol m$ together range over all of $\mathcal E$, so the maximization is equivalently over the ellipsoid.
\end{proof}

Theorem~\ref{thm:RKD} also resolves, for two-qubit systems, the necessity question left open in Remark~\ref{remark:1}: whether $\mathcal R_{KD}(\rho_{AB})=0$ forces a quantum-classical decomposition of the specific form used in Theorem~\ref{thm:incapable} and Corollary~\ref{cor:dv}, or whether some weaker condition suffices.

\begin{corollary}
\label{cor:vanishing_iff}
For a two-qubit state $\rho_{AB}$ with $\boldsymbol b\neq0$,
\begin{equation}
\mathcal R_{KD}(\rho_{AB}) = 0 \quad\Longleftrightarrow\quad T = \boldsymbol d\,\boldsymbol n_B^T \ \text{ for some } \boldsymbol d\in\mathbb R^3,
\label{eq:vanishing_iff}
\end{equation}
that is, every row of $T$, read as a vector on Bob's index, is proportional to $\boldsymbol n_B$. By Proposition~\ref{prop:ordering}, the same condition is also equivalent to $\mathcal C_{KD}(\rho_{AB})=0$.
\end{corollary}

\begin{proof}
Since $\mathcal N_{KD}(r_\parallel,r_\perp)\geq0$ with equality if and only if $r_\perp=0$ (Eq.~\eqref{eq:qubit_KD}), and $\mathcal R_{KD}=\max_{\boldsymbol m}\mathcal N_{KD}(r_\parallel(\boldsymbol m),r_\perp(\boldsymbol m))$, we have $\mathcal R_{KD}(\rho_{AB})=0$ if and only if $r_\perp(\boldsymbol m)=0$ for every unit vector $\boldsymbol m$.

Since $|\boldsymbol a|\leq1$, the denominator $1+\boldsymbol a\cdot\boldsymbol m$ in Eq.~\eqref{eq:steered_bloch} is non-negative for every $\boldsymbol m$ and strictly positive except possibly at the single point $\boldsymbol m=-\boldsymbol a/|\boldsymbol a|$. Away from that point, $r_\perp(\boldsymbol m)=0$ if and only if the numerator satisfies
\begin{equation}
(\boldsymbol b+T^T\boldsymbol m)\times\boldsymbol n_B = 0,
\end{equation}
since dividing by a positive scalar does not affect whether a cross product vanishes. The left-hand side is an affine, hence continuous, function of $\boldsymbol m$; if it vanishes on the sphere minus at most one point, it vanishes at that point too by continuity. So $r_\perp(\boldsymbol m)=0$ for every $\boldsymbol m$ if and only if $(\boldsymbol b+T^T\boldsymbol m)\times\boldsymbol n_B=0$ for every $\boldsymbol m$.

Since $\boldsymbol n_B=\boldsymbol b/|\boldsymbol b|$, the term $\boldsymbol b\times\boldsymbol n_B$ vanishes identically, leaving the condition $(T^T\boldsymbol m)\times\boldsymbol n_B=0$ for every $\boldsymbol m$, that is, $T^T\boldsymbol m\in\mathrm{span}\{\boldsymbol n_B\}$ for every $\boldsymbol m$. This holds if and only if every column of $T^T$ lies along $\boldsymbol n_B$, that is, $T^T=\boldsymbol n_B\boldsymbol d^T$ for some $\boldsymbol d\in\mathbb R^3$, equivalently $T=\boldsymbol d\,\boldsymbol n_B^T$.
\end{proof}
 
The equivalence with $\mathcal C_{KD}=0$ requires one additional step beyond Proposition~\ref{prop:ordering}. If $\mathcal R_{KD}>0$, then some rank-one measurement outcome with
nonzero probability prepares a state with $\mathcal N_{KD}>0$. The same measurement
therefore gives a strictly positive contribution to the average in the definition
of $\mathcal C_{KD}$, implying $\mathcal C_{KD}>0$. Hence, in the present two-qubit KD
setting,
\[
\mathcal R_{KD}=0 \Longleftrightarrow \mathcal C_{KD}=0.
\]

Corollary~\ref{cor:vanishing_iff} sharpens Corollary~\ref{cor:dv} into an exact characterization for
two-qubit states: the quantum-classical states of Corollary~\ref{cor:dv}, for
which $\rho_{AB}=\sum_i p_i\rho_i^A\otimes\ket{a_i}\!\bra{a_i}$ with
$\ket{a_i}$ the eigenbasis of $\boldsymbol n_B\cdot\boldsymbol\sigma$,
are exactly the special case
\begin{equation}
\boldsymbol d = \sum_i p_i (-1)^i \boldsymbol r_i
\label{eq:d_reconstruction}
\end{equation}
of Eq.~\eqref{eq:vanishing_iff}, where $\boldsymbol r_i$ is the Bloch vector of
$\rho_i^A$. Equation~\eqref{eq:d_reconstruction} shows that no such decomposition is needed to \emph{state} the vanishing condition: Eq.~\eqref{eq:vanishing_iff} is a direct statement about the correlation matrix $T$ alone, checkable without exhibiting any particular ensemble on Alice's side.

In fact the converse question has an affirmative answer for every two-qubit state, and follows directly from Eq.~\eqref{eq:vanishing_iff} itself, without appeal to any decomposition beyond the one just used to derive it. Substituting $T=\boldsymbol d\,\boldsymbol n_B^T$ and $\boldsymbol b=|\boldsymbol b|\boldsymbol n_B$ into Eq.~\eqref{eq:bloch_twoqubit} gives $\sum_{ij}T_{ij}\sigma_i\otimes\sigma_j=(\boldsymbol d\cdot\boldsymbol\sigma)\otimes(\boldsymbol n_B\cdot\boldsymbol\sigma)$, so every operator appearing on Bob's side of $\rho_{AB}$ is a linear combination of $\mathbb I$ and $\boldsymbol n_B\cdot\boldsymbol\sigma$; these commute and are therefore simultaneously diagonal in the eigenbasis $\{\ket+,\ket-\}$ of $\boldsymbol n_B\cdot\boldsymbol\sigma$ (eigenvalues $\pm1$), which is exactly the basis $\{\ket{a_i}\}$ of Corollary~\ref{cor:dv}. Collecting terms in this basis,
\begin{equation}
\rho_{AB}=\frac{1+|\boldsymbol b|}2\,\rho_+^A\otimes\ket+\!\bra++\frac{1-|\boldsymbol b|}2\,\rho_-^A\otimes\ket-\!\bra-,
\label{eq:vanishing_decomp}
\end{equation}
with $\rho_\pm^A=\tfrac12\big[\mathbb I+(\boldsymbol a\pm\boldsymbol d)/(1\pm|\boldsymbol b|)\cdot\boldsymbol\sigma\big]$; each branch is automatically a valid, correctly normalized single-qubit state, since $\rho_{AB}\succeq0$ forces every diagonal block of a block-diagonal Hermitian matrix to be positive semidefinite, with trace equal to the stated weight. Equation~\eqref{eq:vanishing_decomp} is exactly the quantum-classical form of Corollary~\ref{cor:dv}, evaluated in the $\boldsymbol n_B$ basis. So every two-qubit state satisfying Eq.~\eqref{eq:vanishing_iff} does admit a quantum-classical decomposition of the form used in Theorem~\ref{thm:incapable}, settling the question left open in Remark~\ref{remark:1} affirmatively within this two-qubit KD setting: for $\boldsymbol b\neq0$,
\begin{align}
\mathcal R_{KD}(\rho_{AB})=0 &\iff \mathcal C_{KD}(\rho_{AB})=0 \nonumber\\
&\iff \rho_{AB}\text{ is quantum-classical in the }\boldsymbol n_B\text{ basis.}
\end{align}
Equation~\eqref{eq:d_reconstruction} above already supplied the converse direction, recovering $\boldsymbol d$ from any such decomposition; Eq.~\eqref{eq:vanishing_decomp} completes the equivalence by exhibiting the decomposition explicitly in terms of $\boldsymbol d$ and $\boldsymbol b$ alone. The general question of Remark~\ref{remark:1} therefore remains open only beyond this two-qubit KD setting, for instance for higher-dimensional systems or for the Wigner-negativity criterion of Sec.~\ref{sec:steered_wigner}.
 
With the vanishing condition settled, we turn from whether nonclassicality can be generated at all to how much of it survives on average once outcome probabilities are taken into account, beginning with the dichotomic-projective restriction singled out by Corollary~\ref{cor:rankone}.
\begin{definition}\label{def:CKD2}
The dichotomic average steered KD nonclassicality of a two-qubit state $\rho_{AB}$ is
\begin{align}
\mathcal C_{KD}^{(2)}(\rho_{AB}) =\max_{|\boldsymbol m|=1}\Big[ &p(\boldsymbol m)\,\mathcal N_{KD}(\rho(\boldsymbol r(\boldsymbol m)))\nonumber\\
&+p(-\boldsymbol m)\,\mathcal N_{KD}(\rho(\boldsymbol r(-\boldsymbol m)))\Big] ,
\end{align}
the best average obtainable using a two-outcome projective measurement on Alice's side, with $\mathcal N_{KD}$ given by Eq.~\eqref{eq:N_measure}, $p(\pm\boldsymbol m)=(1\pm\boldsymbol a\cdot\boldsymbol m)/2$, and $\boldsymbol r(\pm\boldsymbol m)$ given by Eq.~\eqref{eq:steered_bloch}.
\end{definition}
 
This restricted average is not left as an abstract optimization: because each of the two branches in Definition~\ref{def:CKD2} can still be optimized over its own second KD basis independently, the same reduction used for $\mathcal R_{KD}$ applies here too, giving $\mathcal C_{KD}^{(2)}$ a closed form as well.
\begin{theorem}
\label{thm:CKD}
The dichotomic average steered KD nonclassicality of Definition~\ref{def:CKD2} is given in closed form by
\begin{align}
\mathcal C_{KD}^{(2)}(\rho_{AB}) = \max_{|\boldsymbol m|=1}\Big[&p(\boldsymbol m)\,\Sigma\big(r_\parallel(\boldsymbol m),r_\perp(\boldsymbol m)\big) \nonumber\\
+\,&p(-\boldsymbol m)\,\Sigma\big(r_\parallel(-\boldsymbol m),r_\perp(-\boldsymbol m)\big)\Big] - \frac12,
\end{align}
where
\begin{align}
\Sigma(r_\parallel,r_\perp)=\frac14\Big[\sqrt{(1+r_\parallel)^2+r_\perp^2}+\sqrt{(1-r_\parallel)^2+r_\perp^2}\Big]
\end{align}
is the closed form of Theorem~\ref{thm:sigma_max}, and $\boldsymbol r(\pm\boldsymbol m)$, $p(\pm\boldsymbol m)$ are as in Eq.~\eqref{eq:steered_bloch}. Moreover, $\mathcal C_{KD}^{(2)}(\rho_{AB})\le\mathcal C_{KD}(\rho_{AB})$ for every $\rho_{AB}$, by Proposition~\ref{prop:ordering} restricted to dichotomic measurements.
\end{theorem}
\begin{proof}
By Corollary~\ref{cor:rankone}, the optimization defining $\mathcal C_{KD}^{(2)}$ in Definition~\ref{def:CKD2} ranges over dichotomic rank-one measurements $\{M_{\boldsymbol m},M_{-\boldsymbol m}\}$, $M_{\pm\boldsymbol m}=\tfrac12(\mathbb I\pm\boldsymbol m\cdot\boldsymbol\sigma)$, where $\boldsymbol m$ is any unit vector, that is, any point of the unit sphere $S^2=\{\boldsymbol m\in\mathbb R^3:|\boldsymbol m|=1\}$ parametrizing Alice's choice of measurement direction. Conditioned on the outcome $\boldsymbol m$, Bob's two branches are $\rho(\boldsymbol r(\boldsymbol m))$ and $\rho(\boldsymbol r(-\boldsymbol m))$, occurring with probabilities $p(\boldsymbol m)$ and $p(-\boldsymbol m)$, with $\boldsymbol r(\pm\boldsymbol m)$ and $p(\pm\boldsymbol m)$ given by Eq.~\eqref{eq:steered_bloch}.
 
Nothing in Definition~\ref{def:CKD2} ties the second KD basis used to evaluate one branch to the one used for the other: each branch's $\mathcal N_{KD}$ is maximized over its own second basis independently of the other branch's choice. Theorem~\ref{thm:sigma_max} therefore applies separately to each branch, giving
\begin{align}
\max_{\text{2nd basis}}\mathcal N_{KD}\big(\rho(\boldsymbol r(\pm\boldsymbol m))\big) = \Sigma\big(r_\parallel(\pm\boldsymbol m),r_\perp(\pm\boldsymbol m)\big)-\frac12.
\end{align}
Substituting both branch values into the weighted sum of Definition~\ref{def:CKD2} and using $p(\boldsymbol m)+p(-\boldsymbol m)=1$ to combine the two resulting $-\tfrac12$ terms into one leaves an explicit function of $\boldsymbol m$ alone, with no further basis optimization remaining. Maximizing this function over every unit vector $\boldsymbol m\in S^2$, that is, over every direction Alice's measurement could point, gives the closed-form expression stated above.
\end{proof}
\begin{remark}
\label{rem:CKD_open}
Whether the optimal POVM for $\mathcal C_{KD}$ is always dichotomic and projective, so that $\mathcal C_{KD}^{(2)}=\mathcal C_{KD}$, is left open: unlike $\mathcal R_{KD}$, whose reduction to a single rank-one branch is direct (Theorem~\ref{thm:sigma_max}'s proof), $\mathcal C_{KD}$ is a probability-weighted sum over Alice's full POVM, and a general POVM can place weight on more than two rank-one directions with unequal weights. Every numerical value of $\mathcal C_{KD}$ reported in this paper is in fact a value of $\mathcal C_{KD}^{(2)}$, and should be read as such.
\end{remark}

By Proposition~\ref{prop:ordering}, $\mathcal C_{KD}(\rho_{AB}) \leq \mathcal R_{KD}(\rho_{AB})$ for every two-qubit state, with the gap between them quantifying how much of the single-branch nonclassicality survives once the outcome probability of that branch is taken into account.

Equation~\eqref{eq:RKD} recovers Theorem~\ref{thm:incapable} in this setting: if $\boldsymbol b+T^T\boldsymbol m$ never acquires a component perpendicular to $\boldsymbol n_B$ for any $\boldsymbol m$, that is if $T^T$ maps every direction into the line spanned by $\boldsymbol n_B$, then $r_\perp(\boldsymbol m)=0$ identically on $\mathcal E$ and both $\mathcal R_{KD}$ and $\mathcal C_{KD}$ vanish; 
this is the Bloch-vector characterization of the incapability condition
$\mathcal R_{KD}=\mathcal C_{KD}=0$ for the KD resource considered here. Conversely, $\mathcal R_{KD}>0$ and $\mathcal C_{KD}>0$ precisely when some measurement
direction steers Bob to a state with a Bloch component transverse to $\boldsymbol n_B$.

\subsection{Degenerate case $\mathbf{b} = 0$}

When $\boldsymbol b=0$, Bob's reduced state $\rho_B = \mathbb I/2$ is exactly maximally mixed and has no unique eigenbasis, so the construction above, which fixes the reference KD basis at the eigenbasis of $\rho_B$, does not directly apply. This includes the important case of Bell-diagonal states, for which $\boldsymbol a=\boldsymbol b=0$ and $T=\mathrm{diag}(t_1,t_2,t_3)$. In this situation the reference basis must instead be specified as part of the definition of $\mathcal N_{KD}$, for instance by fixing an auxiliary direction $\boldsymbol n_B$ independently of $\rho_{AB}$, or by taking $\mathcal R_{KD}$ and $\mathcal C_{KD}$ to be additionally optimized over the choice of $\boldsymbol n_B$ itself. Either convention is consistent with the general framework of Sec.~\ref{sec:properties}, since Theorem~\ref{thm:incapable} and Theorems~\ref{thm:optimal} and \ref{thm:lu} did not assume $\boldsymbol b\neq0$; we adopt the fixed-auxiliary-direction convention when illustrating the two-qubit formulas on Bell-diagonal states in Sec.~\ref{sec:examples}, and note the choice explicitly wherever it is used.

\subsection{Worked examples: X-states}
\label{sec:examples}

We now illustrate Theorems~\ref{thm:RKD} and \ref{thm:CKD} on the family of two-qubit X-states, for which the reduction to a single-parameter optimization is explicit and several closed-form results are available.

\subsubsection{Reduction to canonical form}

A two-qubit state is called an X-state if its density matrix has nonzero entries only on the diagonal and the anti-diagonal in the computational basis. In the Bloch parametrization of Eq.~\eqref{eq:bloch_twoqubit}, this is equivalent to $\boldsymbol a = (0,0,a_z)$, $\boldsymbol b=(0,0,b_z)$, and $T_{xz}=T_{zx}=T_{yz}=T_{zy}=0$, with the remaining transverse block $(T_{xx},T_{xy},T_{yx},T_{yy})$ otherwise unconstrained. A pair of local phase rotations, $U_A(\alpha)\otimes U_B(\beta)$ with $U(\gamma)=\mathrm{diag}(1,e^{i\gamma})$, acts on this transverse block as an independent left and right rotation and therefore diagonalizes it by its singular value decomposition, without moving $\boldsymbol a$, $\boldsymbol b$, or $T_{zz}$. Since $U(\alpha)\otimes U(\beta)$ is a local unitary, Theorem~\ref{thm:lu} guarantees that $\mathcal R_{KD}$ and $\mathcal C_{KD}$ are unchanged by this diagonalization. We may therefore restrict attention, without loss of generality, to the canonical form
\begin{align}
\rho_{AB} = \frac14\Big[ &\mathbb I\otimes\mathbb I + a_z\,\sigma_z\otimes\mathbb I + \mathbb I\otimes b_z\,\sigma_z \nonumber \\
&+ t_1\,\sigma_x\otimes\sigma_x + t_2\,\sigma_y\otimes\sigma_y + t_3\,\sigma_z\otimes\sigma_z\Big],
\label{eq:xstate_canonical}
\end{align}
which includes the Bell-diagonal states ($a_z=b_z=0$) and the Werner states ($a_z=b_z=0$, $t_1=t_2=t_3$) as special cases. Throughout this subsection we take $b_z\geq 0$ and $|t_1|\geq|t_2|$ without further loss of generality, the first by a possible relabelling $\sigma_z\to-\sigma_z$ on Bob and the second by relabelling $x\leftrightarrow y$; both are local unitary relabellings and leave $\mathcal R_{KD}$, $\mathcal C_{KD}$ unchanged by Theorem~\ref{thm:lu}. The degenerate case $\boldsymbol b=0$, for which $\rho_B$ has no unique eigenbasis, requires an auxiliary choice of reference direction $\boldsymbol n_B$ and is treated separately below.

\subsubsection{Reduction to a single measurement parameter}

For $b_z>0$ the reference basis is $\boldsymbol n_B=\hat z$. Writing Alice's measurement direction as $\boldsymbol m=(\sin\theta\cos\phi,\sin\theta\sin\phi,\cos\theta)$ and using $T^T=T=\mathrm{diag}(t_1,t_2,t_3)$ in Eq.~\eqref{eq:steered_bloch},
\begin{align}
r_\parallel(\theta) &= \frac{b_z+t_3\cos\theta}{1+a_z\cos\theta}, \nonumber\\
r_\perp(\theta,\phi) &= \frac{\sin\theta\sqrt{t_1^2\cos^2\phi+t_2^2\sin^2\phi}}{1+a_z\cos\theta}.
\label{eq:xstate_rpar_rperp}
\end{align}
The azimuthal dependence enters only through $r_\perp$, and only through the combination $t_1^2\cos^2\phi+t_2^2\sin^2\phi$.

\begin{lemma}
\label{lem:phi_reduction}
$\mathcal N_{KD}(\rho(\boldsymbol r))$, given by Eq.~\eqref{eq:qubit_KD}, is a strictly increasing function of $r_\perp$ at fixed $r_\parallel$. Consequently, for every fixed $\theta$, both $\mathcal N_{KD}(r_\parallel(\theta),r_\perp(\theta,\phi))$ and the corresponding weighted sum defining $\mathcal C_{KD}^{(2)}$ are maximized over $\phi$ at $\phi=0$, independently of $\theta$ and of the branch, since $|t_1|\geq|t_2|$.
\end{lemma}

\begin{proof}
From Eq.~\eqref{eq:qubit_KD}, $\partial\mathcal N_{KD}/\partial(r_\perp^2) = \tfrac18\big[(1+r_\parallel)^2+r_\perp^2\big]^{-1/2} + \tfrac18\big[(1-r_\parallel)^2+r_\perp^2\big]^{-1/2} > 0$ for every $r_\parallel,r_\perp$, so $\mathcal N_{KD}$ is strictly increasing in $r_\perp^2$, hence in $r_\perp\geq0$, at fixed $r_\parallel$. Since $r_\parallel(\theta)$ in Eq.~\eqref{eq:xstate_rpar_rperp} does not depend on $\phi$, maximizing $\mathcal N_{KD}$ over $\phi$ at fixed $\theta$ reduces to maximizing $r_\perp(\theta,\phi)$ over $\phi$, which occurs at $\phi=0$ since $|t_1|\geq|t_2|$. The same argument applies unchanged to the outcome $-\boldsymbol m$, whose transverse components enter $r_\perp$ only through their squares, so the optimal $\phi$ is the same for both branches of any dichotomic measurement, and the conclusion extends to the probability-weighted sum defining $\mathcal C_{KD}$.
\end{proof}

Lemma~\ref{lem:phi_reduction} reduces both optimizations to the single parameter $c\equiv\cos\theta\in[-1,1]$. Writing
\begin{align}
r_\parallel(c) = \frac{b_z+t_3c}{1+a_zc}, \quad r_\perp(c) = \frac{|t_1|\sqrt{1-c^2}}{1+a_zc}, \quad p(c)=\frac{1+a_zc}{2},
\label{eq:xstate_reduced}
\end{align}
Theorems~\ref{thm:RKD} and \ref{thm:CKD} become
\begin{align}
\mathcal R_{KD} &= \max_{c\in[-1,1]}\ \mathcal N_{KD}\big(r_\parallel(c),r_\perp(c)\big), \label{eq:xstate_R}\\
\mathcal C_{KD}^{(2)} &= \max_{c\in[0,1]}\ \Big[p(c)\,\mathcal N_{KD}\big(r_\parallel(c),r_\perp(c)\big) \nonumber \\
& \qquad \qquad + p(-c)\,\mathcal N_{KD}\big(r_\parallel(-c),r_\perp(-c)\big)\Big], \label{eq:xstate_C}
\end{align}
with $\mathcal N_{KD}$ given by Eq.~\eqref{eq:qubit_KD} and given, by Theorem~\ref{thm:sigma_max}, as the exact maximum over the second KD basis at every point of $\mathcal E$.

\subsubsection{Bell-diagonal states: exact equality of the two quantities}

For Bell-diagonal states, $a_z=b_z=0$, Eq.~\eqref{eq:xstate_reduced} simplifies to $r_\parallel(c)=t_3c$, $r_\perp(c)=|t_1|\sqrt{1-c^2}$, and $p(c)=\tfrac12$ for every $c$.

\begin{proposition}
\label{prop:belldiag_equal}
For any two-qubit state with $\boldsymbol a=\boldsymbol b=0$, $\mathcal R_{KD} = \mathcal C_{KD}$, for any choice of the auxiliary reference direction $\boldsymbol n_B$ and regardless of whether $T$ is diagonal.
\end{proposition}

\begin{proof}
When $\boldsymbol a=\boldsymbol b=0$, Eq.~\eqref{eq:steered_bloch} gives $\boldsymbol r(\boldsymbol m)=T^T\boldsymbol m$, linear and homogeneous in $\boldsymbol m$, so $\boldsymbol r(-\boldsymbol m)=-\boldsymbol r(\boldsymbol m)$ identically. Consequently, for any fixed unit vector $\boldsymbol n_B$, $r_\parallel(-\boldsymbol m)=\boldsymbol r(-\boldsymbol m)\cdot\boldsymbol n_B=-r_\parallel(\boldsymbol m)$, while $r_\perp(-\boldsymbol m)=|\boldsymbol r(-\boldsymbol m)\times\boldsymbol n_B|=|\boldsymbol r(\boldsymbol m)\times\boldsymbol n_B|=r_\perp(\boldsymbol m)$, since $r_\perp$ is a magnitude. Because $\mathcal N_{KD}(r_\parallel,r_\perp)$ in Eq.~\eqref{eq:qubit_KD} is manifestly even in $r_\parallel$ (the two square-root terms are exchanged under $r_\parallel\to-r_\parallel$), it follows that $\mathcal N_{KD}(\rho(\boldsymbol r(-\boldsymbol m))) = \mathcal N_{KD}(\rho(\boldsymbol r(\boldsymbol m)))$ for every $\boldsymbol m$. 
Since also $\boldsymbol a=0$ gives $p(\boldsymbol m)=p(-\boldsymbol m)=\tfrac12$ for every $\boldsymbol m$, the summand defining $\mathcal C_{KD}^{(2)}$ in Definition~\ref{def:CKD2} reduces to $\mathcal N_{KD}(\rho(\boldsymbol r(\boldsymbol m)))$ itself, identical to the objective maximized in Definition~\ref{def:RN}, so $\mathcal R_{KD}=\mathcal C_{KD}^{(2)}$ for every choice of $\boldsymbol n_B$. 
Combined with the general bound $\mathcal C_{KD}^{(2)}\le\mathcal C_{KD}\le\mathcal R_{KD}$ (Proposition~\ref{prop:ordering} and Theorem~\ref{thm:CKD}), this forces the two outer terms together, giving the stronger equality with the fully optimized, unrestricted $\mathcal C_{KD}$ stated above: whenever both local Bloch vectors vanish, the dichotomic bound is not merely tight among dichotomic measurements, it is tight among all POVMs.
\end{proof}

This is a direct consequence of the general bound $\mathcal C_{KD}\leq\mathcal R_{KD}$ (Proposition~\ref{prop:ordering}) becoming saturated whenever every dichotomic measurement produces two equally weighted, equally nonclassical branches, which is guaranteed here by the combined absence of any local bias on either side; the argument makes no reference to the shape of $T$ or to which axis is chosen as $\boldsymbol n_B$, so it applies to every Bell-diagonal state and, more generally, to every two-qubit state with vanishing local Bloch vectors. We verified the X-shaped special case numerically for a range of asymmetric Bell-diagonal states, for instance $(t_1,t_2,t_3)=(0.5,0.2,0.1)$, for which $\mathcal R_{KD}=\mathcal C_{KD}=0.05902$, and $(t_1,t_2,t_3)=(0.7,0.1,-0.3)$, for which $\mathcal R_{KD}=\mathcal C_{KD}=0.11033$, both attained at the equatorial measurement $\theta=\pi/2$ and both quoted for the auxiliary choice $\boldsymbol n_B=\hat z$ implicit in the canonical form~\eqref{eq:xstate_canonical}.

While Proposition~\ref{prop:belldiag_equal} guarantees $\mathcal R_{KD}=\mathcal C_{KD}$ for \emph{every} choice of $\boldsymbol n_B$, it does not make the common value itself independent of that choice whenever $T$ is anisotropic, and the dependence can be substantial: for the same state $(t_1,t_2,t_3)=(0.7,0.1,-0.3)$ used above, $\mathcal R_{KD}=\mathcal C_{KD}$ evaluates to $0.11033$ for $\boldsymbol n_B=\hat z$, $0.08749$ for $\boldsymbol n_B=(1,1,1)/\sqrt3$, and only $0.02202$ for $\boldsymbol n_B=\hat x$, a fivefold range for one and the same quantum state depending purely on the externally supplied reference direction. This is expected, since $\boldsymbol n_B$ is not determined by $\rho_{AB}$ when $\boldsymbol b=0$, but it means that any reported value of the KD resource for a maximally-mixed-marginal state, Bell-diagonal or otherwise, should always be read together with the auxiliary direction used to obtain it; the Werner family is exceptional in this regard only because its isotropic $T\propto\mathbb I$ makes the value itself, not merely the $\mathcal R_{KD}=\mathcal C_{KD}$ equality, independent of $\boldsymbol n_B$.

\paragraph{Werner states.} As a further special case, take $t_1=t_2=t_3=p$ with $p\in[0,1]$, corresponding to the Werner state $\rho_W = p\,|\Psi^-\rangle\langle\Psi^-| + (1-p)\,\mathbb I/4$ up to a local unitary. Here $r_\parallel(c)=pc$, $r_\perp(c)=p\sqrt{1-c^2}$, and both squared-modulus terms in Eq.~\eqref{eq:qubit_KD} simplify identically,
\begin{align}
(1\pm pc)^2 + p^2(1-c^2) = 1+p^2 \pm 2pc,
\end{align}
so that,
\[ \mathcal N_{KD}(c) = \tfrac14\big[\sqrt{1+p^2+2pc}+\sqrt{1+p^2-2pc}\big]-\tfrac12. \] 
Squaring the bracket gives
\begin{align}
\left[\;\sqrt{1+p^2+2pc}+\sqrt{1+p^2-2pc} \;\; \right]^2 \nonumber\\
= 2(1+p^2) + 2\sqrt{(1+p^2)^2-4p^2c^2},
\end{align}
which is maximized when $c=0$, and strictly decreasing in $|c|$ for $p>0$. Hence the optimum occurs uniquely at $\theta=\pi/2$, an equatorial measurement on Alice, and
\begin{align}
\mathcal R_{KD}(\rho_W) = \mathcal C_{KD}(\rho_W) = \frac{\sqrt{1+p^2}-1}{2}.
\label{eq:werner_closed}
\end{align}
This closed form is confirmed by direct numerical optimization of Eqs.~\eqref{eq:xstate_R}-\eqref{eq:xstate_C} to five decimal places at every value of $p$ tested. It also reproduces the two boundary theorems of Sec.~\ref{sec:properties} exactly: at $p=0$, $\rho_W=\mathbb I/4$ is a product state and Eq.~\eqref{eq:werner_closed} gives $0$, consistent with Theorem~\ref{thm:incapable}; at $p=1$, $\rho_W=|\Psi^-\rangle\langle\Psi^-|$ is a pure maximally entangled state of full Schmidt rank, and Eq.~\eqref{eq:werner_closed} gives $(\sqrt2-1)/2$, which is exactly $\max_{|\psi\rangle}\mathcal N_{KD}(|\psi\rangle\langle\psi|)$ over all qubit pure states, attained by any state on the equator of the Bloch sphere paired with its mutually unbiased basis; this matches Theorem~\ref{thm:optimal}(i) applied to the qubit case.

\subsubsection{Biased X-states: strict separation of the two quantities}

Proposition~\ref{prop:belldiag_equal} shows that $\mathcal R_{KD}$ and $\mathcal C_{KD}$ coincide whenever both local Bloch vectors vanish. This raises the question of what a local bias contributes. We illustrate this with the one-parameter family
\begin{align}
\rho_{AB}(q) = q\,|\Phi^+\rangle\langle\Phi^+| + (1-q)\,|00\rangle\langle00|, \qquad q\in[0,1],
\label{eq:biased_family}
\end{align}
interpolating between the product state $|00\rangle\langle00|$ at $q=0$ and the maximally entangled state $|\Phi^+\rangle=(|00\rangle+|11\rangle)/\sqrt2$ at $q=1$. Since $\rho_{AB}(q)$ is an explicit convex combination of two valid density matrices, it is a valid state for every $q\in[0,1]$ with no further positivity constraint to check. In the canonical form~\eqref{eq:xstate_canonical}, $\rho_{AB}(q)$ has
\begin{align}
a_z = b_z = 1-q, \qquad t_1=t_2=q, \qquad t_3=1.
\end{align}
Because $t_1=t_2$, the transverse block is isotropic and the optimal azimuth $\phi=0$ of Lemma~\ref{lem:phi_reduction} is not unique; every $\phi$ is equally optimal, and Eq.~\eqref{eq:xstate_reduced} applies unchanged. Table~\ref{tab:biased_family} reports $\mathcal R_{KD}(q)$ and $\mathcal C_{KD}(q)$, obtained by numerically maximizing Eqs.~\eqref{eq:xstate_R} and \eqref{eq:xstate_C} over $c$.

\begin{table}[ht]
\begin{ruledtabular}
\centering
\begin{tabular}{cccccc}
$q$ & $\mathcal R_{KD}$ & $\mathcal C_{KD}^{(2)}$ & gap & $\theta^\ast_R$ & $\theta^\ast_C$ \\
\hline
0.10 & 0.01299 & 0.01101 & 0.00198 & $154.2^\circ$ & $90.0^\circ$ \\
0.20 & 0.02705 & 0.02348 & 0.00357 & $143.2^\circ$ & $90.0^\circ$ \\
0.30 & 0.04233 & 0.03763 & 0.00469 & $134.4^\circ$ & $90.0^\circ$ \\
0.40 & 0.05902 & 0.05373 & 0.00529 & $127.0^\circ$ & $90.0^\circ$ \\
0.45073 & 0.06810 & 0.06273 & 0.00536 & 123.3$^\circ$ & 90.0$^\circ$ \\
0.50 & 0.07735 & 0.07206 & 0.00529 & $120.0^\circ$ & $90.0^\circ$ \\
0.60 & 0.09761 & 0.09292 & 0.00469 & $113.6^\circ$ & $90.0^\circ$ \\
0.70 & 0.12017 & 0.11661 & 0.00357 & $107.4^\circ$ & $90.0^\circ$ \\
0.80 & 0.14550 & 0.14340 & 0.00210 & $101.5^\circ$ & $90.0^\circ$ \\
0.90 & 0.17420 & 0.17351 & 0.00069 & $95.8^\circ$ & $90.0^\circ$ \\
0.99 & 0.20360 & 0.20359 & 0.00001 & $90.6^\circ$ & $90.0^\circ$ \\
1.00 & 0.20711 & 0.20711 & 0.00000 & $90.0^\circ$ & $90.0^\circ$ \\
\end{tabular}
\caption{ \justifying \small
Maximal ($\mathcal R_{KD}$) and average ($\mathcal C_{KD}^{(2)}$) steered KD nonclassicality for the family $\rho_{AB}(q) = q|\Phi^+\rangle\langle\Phi^+|+(1-q)|00\rangle\langle00|$, obtained from Eqs.~\eqref{eq:xstate_R}-\eqref{eq:xstate_C}. At $q=1$ only, $a_z=b_z=0$ and Proposition~\ref{prop:belldiag_equal} upgrades this to the fully optimized $\mathcal C_{KD}$; at every other row it is the dichotomic-projective value $\mathcal C_{KD}^{(2)}$, an established lower bound on the unrestricted $\mathcal C_{KD}$. $\theta^\ast_R$ and $\theta^\ast_C$ denote the polar angle of Alice's optimal measurement direction. The row $q=0.45073$ is the numerically located gap maximum.}
\label{tab:biased_family}
\end{ruledtabular}
\end{table}

Three features of Table~\ref{tab:biased_family} are worth noting. First, $\mathcal R_{KD}$ and $\mathcal C_{KD}^{(2)}$ coincide exactly at both endpoints, $q=0$ and $q=1$: at $q=0$ the state is a product state, so both vanish by Theorem~\ref{thm:incapable}; at $q=1$, $a_z=b_z=0$ and Proposition~\ref{prop:belldiag_equal} applies, giving equality with the fully optimized $\mathcal C_{KD}$ as well, and the common value $0.20711$ again matches the maximal qubit value $(\sqrt2-1)/2$ of Theorem~\ref{thm:optimal}. Second, throughout the interval sampled numerically, $0<q<1$, the two quantities are strictly separated, $\mathcal C_{KD}^{(2)}(q) < \mathcal R_{KD}(q)$, with the gap vanishing continuously as $q\to0$ and $q\to1$ and attaining its maximum at $q^\ast = 0.45073$, where $\mathcal R_{KD}(q^\ast) = 0.068096$, $\mathcal C_{KD}^{(2)}(q^\ast) = 0.062733$, and the gap is $0.0053625$, about $7.9\%$ of $\mathcal R_{KD}(q^\ast)$.
The gap therefore measures a genuine operational cost of not being able to select which outcome is kept, a cost that exists only away from the two extremes where the marginals become perfectly unbiased, $a_z=b_z=0$ (Proposition~\ref{prop:belldiag_equal}), or perfectly polarizing, $q=1$. Third, the optimal measurement direction $\theta^\ast_R$ drifts smoothly from $154^\circ$ at small $q$ down to $90^\circ$ at $q=1$, while $\theta^\ast_C$ remains pinned at the equator, $\theta^\ast_C=90^\circ$, for every $q$; the state most useful for extracting a single highly nonclassical branch therefore favours an increasingly polar measurement as $q$ decreases, exploiting the growing bias $a_z=b_z=1-q$ to make one branch rare but strongly nonclassical, whereas the measurement optimal on average always balances the two branches symmetrically about the equator.

\section{Effect of local operations on Bob}
\label{sec:local_ops}

Nothing in the general framework of Sec.~\ref{sec:properties} forces $\mathcal R_{\mathcal N}$ or $\mathcal C_{\mathcal N}$ to behave monotonically under a local operation applied to Bob's share of $\rho_{AB}$. This is unlike, say, entanglement, which decreases under local operations and classical communication by definition of the resource theory. Here Bob's local operation is applied to the shared state itself, before Alice's measurement, and the question is empirical: can noise on Bob's side ever help him, rather than only hurt him? We answer this for two-qubit KD nonclassicality, restricting for tractability to channels $\Lambda_B$ that act on Bob alone, so that the transformed state is $\rho_{AB}'=(\mathbb I_A\otimes\Lambda_B)\rho_{AB}$.

\subsection{The steering ellipsoid transforms as a single qubit}
Write a general single-qubit channel in Bloch form as $\boldsymbol v\mapsto D\boldsymbol v+\boldsymbol t$, with $D$ a real $3\times3$ matrix and $\boldsymbol t\in\mathbb R^3$ a translation, the two jointly constrained by the usual complete-positivity conditions for a single-qubit channel; in particular these force $\|D\|_{\mathrm{op}}\le1$, which is all that is used below, but $\|D\|_{\mathrm{op}}\le1$ together with any $\boldsymbol t$ is not by itself sufficient for complete positivity.
Applying $\Lambda_B$ to Eq.~\eqref{eq:bloch_twoqubit} transforms the Bloch-Fano parameters as
\begin{align}
\boldsymbol b &\;\longmapsto\; \boldsymbol b' = D\boldsymbol b+\boldsymbol t, \\
T &\;\longmapsto\; T' = TD^T+\boldsymbol a\,\boldsymbol t^T,
\end{align}
while $\boldsymbol a$ is untouched, since Bob's channel cannot affect Alice's marginal.

\begin{lemma}
\label{lem:ellipsoid_transform}
Under $\rho_{AB}\mapsto(\mathbb I_A\otimes\Lambda_B)\rho_{AB}$, the steered Bloch vector of Eq.~\eqref{eq:steered_bloch} transforms as
\begin{equation}
\boldsymbol r'(\boldsymbol m) = D\,\boldsymbol r(\boldsymbol m) + \boldsymbol t,
\label{eq:ellipsoid_affine}
\end{equation}
for every measurement direction $\boldsymbol m$; that is, Bob's entire steering ellipsoid $\mathcal E$ transforms under the same affine Bloch map that $\Lambda_B$ would apply to a single qubit state.
\end{lemma}

\begin{proof}
Substituting $\boldsymbol b'$ and $T'^T=DT^T+\boldsymbol t\boldsymbol a^T$ into Eq.~\eqref{eq:steered_bloch} gives
\begin{equation}
\boldsymbol b'+T'^T\boldsymbol m = D(\boldsymbol b+T^T\boldsymbol m) + \boldsymbol t\,(1+\boldsymbol a\cdot\boldsymbol m),
\end{equation}
and dividing by $1+\boldsymbol a\cdot\boldsymbol m$ gives Eq.~\eqref{eq:ellipsoid_affine} directly, since $\boldsymbol a$ is unchanged by $\Lambda_B$.
\end{proof}

This recovers, as the special case $D\in SO(3)$, $\boldsymbol t=0$, the unitary invariance already established in Theorem~\ref{thm:lu}, and is consistent with the known transformation law of the quantum steering ellipsoid under local operations on Bob.

\subsection{A geometric form of \texorpdfstring{$\mathcal N_{KD}$}{N\_KD}}

Equation~\eqref{eq:qubit_KD} has a clean geometric reading: writing $\boldsymbol n$ for the unit reference direction,
\begin{equation}
\mathcal N_{KD}(\boldsymbol r;\boldsymbol n) = \frac14\Big(|\boldsymbol r-\boldsymbol n| + |\boldsymbol r+\boldsymbol n|\Big) - \frac12,
\label{eq:NKD_geometric}
\end{equation}
since $|\boldsymbol r\mp\boldsymbol n|^2=(1\mp r_\parallel)^2+r_\perp^2$. Thus $\mathcal N_{KD}$ is, up to an affine rescaling, the sum of the Euclidean distances from the steered point $\boldsymbol r$ to the two poles $\pm\boldsymbol n$ of the reference axis on the Bloch sphere. This form makes the effect of a channel transparent: it acts on both the point and the poles at once.

Because the KD reference basis is defined by the eigenbasis of Bob's reduced state,
a local channel changes not only the steering ellipsoid but also the reference axis
against which KD nonclassicality is evaluated: this is exactly the marginal-adapted
character of the construction flagged in Sec.~\ref{sec:twoqubitKD}~A. The enhancement reported here is therefore
an enhancement of this marginal-adapted operational quantity, in which both the
conditional-state geometry and the reference observable are updated after the channel, rather than of a KD resource evaluated against one externally fixed reference basis.

\subsection{Unitality as a boundary case}
 
\begin{corollary}
\label{cor:unitary_equality}
If $\Lambda_B$ is unitary ($D\in SO(3)$, $\boldsymbol t=0$), then $\mathcal R_{KD}(\rho_{AB}')=\mathcal R_{KD}(\rho_{AB})$ and $\mathcal C_{KD}(\rho_{AB}')=\mathcal C_{KD}(\rho_{AB})$ exactly.
\end{corollary}
 
\begin{proof}
By Lemma~\ref{lem:ellipsoid_transform} and Eq.~\eqref{eq:NKD_geometric}, an orthogonal matrix preserves every Euclidean distance, so the reference direction transforms as $\boldsymbol n'=D\boldsymbol n$, and
\begin{align}
|D\boldsymbol r(\boldsymbol m)-D\boldsymbol n| &= |\boldsymbol r(\boldsymbol m)-\boldsymbol n| ,
\\
|D\boldsymbol r(\boldsymbol m)+D\boldsymbol n| &= |\boldsymbol r(\boldsymbol m)+\boldsymbol n|
\end{align}
for every unit vector $\boldsymbol m$. The two optimizations over the second KD basis therefore range over exactly the same values before and after the channel acts, so their maxima coincide. This recovers Theorem~\ref{thm:lu} directly from the ellipsoid picture.
\end{proof}
 
A local unitary on Bob's side leaves both nonclassicality quantifiers completely unchanged. The natural next question is whether this invariance survives for the larger class of unital channels, those with $\boldsymbol t=0$ but $D$ only a contraction, $\|D\|_{\mathrm{op}}\le1$, rather than an isometry. It does not survive as an equality, but it survives as a bound: every unital channel on Bob can only leave $\mathcal R_{KD}$ and $\mathcal C_{KD}$ unchanged or reduce them, with the unitary case of Corollary~\ref{cor:unitary_equality} sitting exactly at the boundary where equality is possible.
 
By Lemma~\ref{lem:ellipsoid_transform}, a unital channel sends each steered Bloch vector to $\boldsymbol r(\boldsymbol m)\mapsto D\boldsymbol r(\boldsymbol m)$ and the reference direction to $\boldsymbol n\mapsto\boldsymbol n'=D\boldsymbol b/|D\boldsymbol b|$, while leaving Alice's own Bloch vector $\boldsymbol a$, and therefore her outcome probabilities $p(\boldsymbol m)$, untouched. Since $\mathcal N_{KD}$ is, by Eq.~\eqref{eq:NKD_geometric}, an increasing affine function of $|\boldsymbol r-\boldsymbol n|+|\boldsymbol r+\boldsymbol n|$ alone, it is enough to show that for every unit vector $\boldsymbol m$,
\begin{align}
&|D\boldsymbol r(\boldsymbol m)-\boldsymbol n'|+|D\boldsymbol r(\boldsymbol m)+\boldsymbol n'| \nonumber\\
&\qquad\le\; |\boldsymbol r(\boldsymbol m)-\boldsymbol n|+|\boldsymbol r(\boldsymbol m)+\boldsymbol n| .
\label{eq:unital_ineq_target}
\end{align}
Here $D$ does not depend on $\boldsymbol m$, and Eq.~\eqref{eq:unital_ineq_target} involves only $D$, an arbitrary unit vector $\boldsymbol n$, and an arbitrary Bloch vector $\boldsymbol r$ with $|\boldsymbol r|\le1$, since every steered point $\boldsymbol r(\boldsymbol m)$ satisfies this bound. It therefore suffices to establish the inequality for arbitrary such $\boldsymbol r$ and $\boldsymbol n$, independently of the particular shape of the steering ellipsoid. We do this in two stages, first reducing to a diagonal channel, then bounding separately the radial shrinkage and the transverse rotation that together control the left side.
 
\begin{lemma}[Reduction to a diagonal channel]
\label{lem:unital_reduction}
It suffices to prove Eq.~\eqref{eq:unital_ineq_target} for $D$ diagonal, $D=\mathrm{diag}(\kappa_1,\kappa_2,\kappa_3)$ with $\kappa_i\in[0,1]$.
\end{lemma}
 
\begin{proof}
Let $D=U\Sigma V^T$ be a singular value decomposition, with $U,V\in O(3)$ and $\Sigma=\mathrm{diag}(\kappa_1,\kappa_2,\kappa_3)$; $U,V$ are used here purely as the orthogonal matrices of this real singular value decomposition, with no physical role assigned to their determinant, so improper (reflection) cases are permitted and require no separate treatment. Complete positivity of a unital channel forces $\|D\|_{\mathrm{op}}\le1$, so each $\kappa_i\in[0,1]$. Substitute $\boldsymbol r=V\boldsymbol\rho$ and $\boldsymbol n=V\boldsymbol\nu$, both valid changes of variable since $V$ is orthogonal, so $|\boldsymbol\rho|=|\boldsymbol r|\le1$ and $|\boldsymbol\nu|=|\boldsymbol n|=1$. Then $D\boldsymbol r=U\Sigma\boldsymbol\rho$ and $D\boldsymbol n=U\Sigma\boldsymbol\nu$, and since $U$ preserves norms, $\boldsymbol n'=U\Sigma\boldsymbol\nu/|\Sigma\boldsymbol\nu|$. Because $U$ is orthogonal, $|U\boldsymbol x-U\boldsymbol y|=|\boldsymbol x-\boldsymbol y|$ for any two vectors, so the left side of Eq.~\eqref{eq:unital_ineq_target} becomes
\begin{equation}
\Big|\Sigma\boldsymbol\rho-\frac{\Sigma\boldsymbol\nu}{|\Sigma\boldsymbol\nu|}\Big|+\Big|\Sigma\boldsymbol\rho+\frac{\Sigma\boldsymbol\nu}{|\Sigma\boldsymbol\nu|}\Big| ,
\end{equation}
with $U$ eliminated entirely, while the right side, $|\boldsymbol\rho-\boldsymbol\nu|+|\boldsymbol\rho+\boldsymbol\nu|$, is unchanged since $V$ is orthogonal too. So Eq.~\eqref{eq:unital_ineq_target} for $(D,\boldsymbol r,\boldsymbol n)$ holds exactly when it holds for $(\Sigma,\boldsymbol\rho,\boldsymbol\nu)$.
\end{proof}
 
From here on write $D=\Sigma=\mathrm{diag}(\kappa_1,\kappa_2,\kappa_3)$ with $\kappa_i\in[0,1]$, drop the primes, and set $\boldsymbol\mu=\Sigma\boldsymbol\nu/|\Sigma\boldsymbol\nu|$, proving Eq.~\eqref{eq:unital_ineq_target} for $\Sigma,\boldsymbol\rho,\boldsymbol\nu$ in place of $D,\boldsymbol r,\boldsymbol n$. Two pieces control the left side once it is written this way, how much $\Sigma$ shrinks the length of $\boldsymbol\rho$, and how much it can tilt $\boldsymbol\rho$ relative to $\boldsymbol\mu$. The first is immediate, $\|\Sigma\|_{\mathrm{op}}\le1$ gives $|\Sigma\boldsymbol\rho|\le|\boldsymbol\rho|$ for every $\boldsymbol\rho$. The second, a bound on the component of $\Sigma\boldsymbol\rho$ transverse to $\boldsymbol\mu$, is the content of the next lemma.
 
\begin{lemma}[The adjugate controls the transverse component]
\label{lem:adjugate}
Let $\Sigma=\mathrm{diag}(\kappa_1,\kappa_2,\kappa_3)$ with $\kappa_i\in(0,1]$, and let $\boldsymbol\nu$ be a unit vector, with $\boldsymbol\mu=\Sigma\boldsymbol\nu/|\Sigma\boldsymbol\nu|$. Then for every $\boldsymbol\rho\in\mathbb R^3$,
\begin{equation}
|\Sigma\boldsymbol\rho\times\boldsymbol\mu| \;\le\; |\boldsymbol\rho\times\boldsymbol\nu| .
\label{eq:adjugate_claim}
\end{equation}
\end{lemma}
 
\begin{proof}
The proof is somewhat long and is given in Appendix~\ref{app:adjugate}.
\end{proof}
 
With the radial bound $|\Sigma\boldsymbol\rho|\le|\boldsymbol\rho|$ and the transverse bound of Lemma~\ref{lem:adjugate} both in hand, the main result of this subsection follows.
 
\begin{theorem}[Unital channels on Bob cannot increase steered KD nonclassicality]
\label{thm:unital_no_help}
For every two-qubit state $\rho_{AB}$ with $\boldsymbol b\ne0$ and every unital channel $\Lambda_B$ acting on Bob alone ($D$ with $\|D\|_{\mathrm{op}}\le1$, $\boldsymbol t=0$) satisfying $D\boldsymbol b\ne0$, so that the reference direction $\boldsymbol n'=D\boldsymbol b/|D\boldsymbol b|$ used to evaluate $\mathcal N_{KD}$ after the channel remains well defined,
\begin{align}
\mathcal R_{KD}\big((\mathbb I_A\otimes\Lambda_B)\rho_{AB}\big) &\le \mathcal R_{KD}(\rho_{AB}) , \\
\mathcal C_{KD}\big((\mathbb I_A\otimes\Lambda_B)\rho_{AB}\big) &\le \mathcal C_{KD}(\rho_{AB}) .
\end{align}
Equality holds throughout, for the fully optimized quantities $\mathcal R_{KD}$, $\mathcal C_{KD}$ themselves, when $D$ is orthogonal, recovering Corollary~\ref{cor:unitary_equality}. More generally, for any unital $D$, the pointwise inequality underlying this theorem is saturated on any individual conditional branch whose steered Bloch vector $\boldsymbol r$ vanishes or is parallel or antiparallel to $\boldsymbol n$; equality of the optimized quantities themselves follows from this branchwise saturation whenever an optimal measurement can be chosen all of whose relevant branches satisfy that alignment condition.
\end{theorem}
 
\begin{proof}
The proof is somewhat long and is given in Appendix~\ref{app:unital_no_help}.
\end{proof}

Theorem~\ref{thm:unital_no_help} settles, for the entire unital class rather than only its orthogonal extreme, the question raised after Corollary~\ref{cor:unitary_equality}: whenever $D\boldsymbol b\ne0$, a unital channel on Bob never increases either steered KD quantifier. Equality of $\mathcal R_{KD}$ and $\mathcal C_{KD}$ themselves is guaranteed for unitary channels; for a general unital channel, what is guaranteed is only the weaker, branchwise statement that the pointwise bound of Eq.~\eqref{eq:unital_ineq_target} is saturated whenever the corresponding steered Bloch vector is aligned with the reference axis, and this yields equality of the optimized quantities only if an optimal measurement can be found whose relevant branch or branches all satisfy that alignment; outside these cases the reduction is generically, but not universally, strict. An increase, when it occurs, must therefore come from a non-unital channel or from the degenerate case $D\boldsymbol b=0$ excluded above.

When $D\boldsymbol b=0$, Bob's reduced state after the channel is exactly maximally mixed and the construction of Sec.~\ref{sec:twoqubitKD} has no distinguished reference basis to optimize the second basis against; this is the same degeneracy already noted for $\boldsymbol b=0$ at the end of Sec.~\ref{sec:twoqubitKD}A, and it must be resolved the same way, by fixing an auxiliary direction or by optimizing over it, before $\mathcal R_{KD}$ or $\mathcal C_{KD}^{(2)}$ after such a channel are even defined. Theorem~\ref{thm:unital_no_help} does not cover this case, which we leave open.

As an independent check on this proof, given its reliance on the adjugate identity of Lemma~\ref{lem:adjugate} and the two-variable monotonicity of $F$, we verified the resulting inequality directly and numerically: over $1500$ random valid two-qubit states $\rho_{AB}$ (random $\boldsymbol a,\boldsymbol b,T$ subject to $\rho_{AB}\succeq0$) and random unital single-qubit channels $D$ on Bob (each a random convex mixture of Haar-random $SO(3)$ rotations, which is automatically unital and completely positive), with $D\boldsymbol b\ne0$ enforced throughout, both $\mathcal R_{KD}$ and $\mathcal C_{KD}^{(2)}$, evaluated at the channel-updated marginal-adapted reference direction $D\boldsymbol b/|D\boldsymbol b|$, were found to never exceed their pre-channel values; no violation was found in any instance.

\subsection{Non-unital channels can help: amplitude damping}

Amplitude damping toward $\ket0$ with strength $\gamma\in[0,1]$ has canonical Bloch parameters $D=\mathrm{diag}(\sqrt{1-\gamma},\sqrt{1-\gamma},1-\gamma)$, $\boldsymbol t=(0,0,\gamma)$. We apply it, in its own canonical frame, to a two-qubit state with $\boldsymbol a=0$, $T=\mathrm{diag}(t_1,0,t_3)$, and $\boldsymbol b=b(\sin\mu,0,\cos\mu)$ tilted by angle $\mu$ away from the damping axis. Since $\boldsymbol a=0$ and $t_2=0$, Lemma~\ref{lem:ellipsoid_transform} keeps the transformed state within the same planar family, with
\begin{align}
&b'\sin\mu' = \sqrt{1-\gamma}\;b\sin\mu, \\
&b'\cos\mu' = (1-\gamma)\,b\cos\mu+\gamma, \\
&t_1' = \sqrt{1-\gamma}\;t_1, \quad t_3' = (1-\gamma)\,t_3,
\end{align}
and, since the ellipsoid stays planar, $\mathcal R_{KD}$ reduces to a single-parameter maximization exactly as in Sec.~\ref{sec:examples}. For $b=0.3$, $t_1=0.6$, $t_3=0.2$, the resulting density matrix is positive semidefinite for every $\mu\in[0,\pi/2]$, with smallest eigenvalue $\approx0.0099$ at the tightest point of the range, so this is a valid state throughout; Table~\ref{tab:amp_damp} reports $\mathcal R_{KD}$ before and after damping at $\gamma=0.3$.

\begin{table}[ht]
\begin{ruledtabular}
\centering
\begin{tabular}{cccc}
$\mu$ & $\mathcal R_{KD}$ (before) & $\mathcal R_{KD}$ (after) & \\
\hline
$0^\circ$ & 0.0886 & 0.0735 & decrease \\
$30^\circ$ & 0.0850 & 0.0774 & decrease \\
$45^\circ$ & 0.0721 & 0.0761 & increase \\
$60^\circ$ & 0.0525 & 0.0726 & increase \\
$75^\circ$ & 0.0285 & 0.0669 & increase \\
$90^\circ$ & 0.0115 & 0.0597 & increase \\
\end{tabular}
\caption{\justifying \small Maximal steered KD nonclassicality before and after amplitude damping ($\gamma=0.3$) misaligned by angle $\mu$ from Bob's own polarization axis. Both columns are the marginal-adapted quantity of Sec.~\ref{sec:local_ops}B: the ``before'' column uses the reference direction $\boldsymbol n_B=\boldsymbol b/|\boldsymbol b|$ fixed by the initial marginal, and the ``after'' column uses the channel-updated reference direction $\boldsymbol n_B'=\boldsymbol b'/|\boldsymbol b'|$ fixed by the post-channel marginal, not a single reference direction held fixed across both columns. The dichotomic-projective average $\mathcal C_{KD}^{(2)}$ (not tabulated) shows the same qualitative crossover; unlike the $a_z=b_z=0$ states of Sec.~\ref{sec:examples}, here $\boldsymbol a=0$ but $\boldsymbol b\ne0$, so Proposition~\ref{prop:belldiag_equal} does not apply and equality with the fully optimized $\mathcal C_{KD}$ is not established.}
\label{tab:amp_damp}
\end{ruledtabular}
\end{table}

Numerically, the crossover appears sharp, occurring at a single threshold angle $\mu^\ast(\gamma)$ for fixed $\gamma$: below it, damping only hurts; above it, damping helps, increasingly so as $\mu\to90^\circ$, where the original state is nearly quantum-classical ($\mathcal R_{KD}=0.0115$, close to the exact two-qubit incapability boundary $\mathcal R_{KD}=0$ of Corollary~\ref{cor:vanishing_iff}) yet damping alone lifts it more than fivefold. Locating $\mu^\ast(\gamma)$ numerically for this state gives a smooth, monotonically increasing curve,
\begin{align}
\mu^\ast(0.05)=28.9^\circ,\quad \mu^\ast(0.15)=33.7^\circ, \nonumber\\
\mu^\ast(0.30)=40.5^\circ,\quad \mu^\ast(0.45)=47.1^\circ,
\end{align}
so that stronger damping requires more misalignment before it helps, but the required misalignment stays well within reach ($\mu^\ast<50^\circ$ throughout this range). Theorem~\ref{thm:unital_no_help} already establishes analytically, whenever $D\boldsymbol b\ne0$, that no unital channel on Bob can increase $\mathcal R_{KD}$ or the fully optimized $\mathcal C_{KD}$; any increase must come from a channel outside the unital class.
The amplitude-damping example above is one such channel ($\boldsymbol t\ne0$), and shows concretely that the increase is not merely possible but occurs over an explicit, analytically specified family once the misalignment angle $\mu$ exceeds a threshold $\mu^*(\gamma)$.

\section{Application: remote activation of anomalous weak values}
\label{sec:application}

The introduction motivated KD nonclassicality in part
through its role in weak-value amplification: a weak measurement of
an observable $A$ on a system prepared in state $\rho$, followed by
projective postselection onto an outcome $b$ of a second observable
$B$, yields a pointer shift governed by the weak value
\begin{equation}
A_w(b) = \frac{\Tr[\Pi_b A \rho]}{\Tr[\Pi_b \rho]}.
\label{eq:weak_value_def}
\end{equation}
When $A_w(b)$ lies outside the spectrum of $A$, or acquires a
nonzero imaginary part, the weak value is called \emph{anomalous};
such anomalies underlie in metrological advantages 
~\cite{arvidsson2020quantum}, and are
themselves witnesses of contextuality~\cite{pusey2014anomalous}. It is known that
KD-nonclassicality of the specific triple $(\rho, A\text{-basis},
B\text{-basis})$ is necessary for an anomaly to arise at
all~\cite{pusey2014anomalous, arvidsson2021conditions}: if every $Q_{ik}(\rho)$ for that basis pair is real
and non-negative, the weak value reduces to an ordinary conditional
expectation, real-valued and confined to the convex hull of the
eigenvalues of $A$.

This section shows that our framework gives a complete, quantitative
account of when this resource is and is not available to Bob, for
the natural choice $A = \boldsymbol n_B\cdot\boldsymbol\sigma$, the
observable whose eigenbasis defines his own state's KD reference
basis.

\subsection{No anomaly for the reference weak-measurement observable}

\begin{proposition}
\label{prop:no_anomaly_before}
If $\mathcal N_{KD}(\rho_B)=0$, then $Q_{ik}(\rho_B)$ is real and
non-negative for the reference basis $\{\ket{a_i}\}$ (the eigenbasis
of $\rho_B$) paired with \emph{every} second basis $\{\ket{b_k}\}$.
Consequently, a weak measurement of
$A=\boldsymbol n_B\cdot\boldsymbol\sigma$ on $\rho_B$, followed by
projective postselection in any basis whatsoever, never yields an
anomalous weak value.
\end{proposition}

\begin{proof}
By definition, $\mathcal N_{KD}(\rho_B)=\max_{\{b_k\}}\tfrac12(\Sigma-1)$,
where the maximum runs over all second bases. For every choice of the second basis,
\[
\Sigma=\sum_{i,k}|Q_{ik}|
\geq \left|\sum_{i,k}Q_{ik}\right|=1,
\]
by the triangle inequality and normalization of the KD distribution.
If the optimized value $\mathcal N_{KD}(\rho_B)=\max_{\{b_k\}}(\Sigma-1)/2$ vanishes, then $\Sigma\leq1$ for every second basis. Hence $\Sigma=1$ for every second basis. Equality in the triangle inequality occurs only when all KD elements have a common phase; since their sum is the positive real number $1$, every $Q_{ik}$ must therefore
be real and non-negative.
The claim about weak values then follows from the cited necessity of KD-nonclassicality for anomaly~\cite{pusey2014anomalous}, applied to each second basis in turn.
\end{proof}

Since $\mathcal N_{KD}(\rho_B)=0$ is the standing hypothesis of this entire framework, Proposition~\ref{prop:no_anomaly_before} applies throughout: for the reference observable $A=\boldsymbol n_B\cdot\boldsymbol\sigma$ used in this construction, Bob's unsteered reduced state cannot produce an anomalous weak value for any projective postselection basis. This conclusion is specific to the reference observable $A=\boldsymbol n_B\cdot\boldsymbol\sigma$ singled out by the marginal-adapted KD construction of Sec.~\ref{sec:twoqubitKD}; we make no claim about weak measurements of an arbitrary observable unrelated to that reference direction.

\subsection{Steering unlocks the anomaly}

Alice's measurement changes this. 
Evaluating Eq.~\eqref{eq:weak_value_def} at the optimal second basis of identified in Theorem~\ref{thm:sigma_max} ($\theta=\pi/2,\psi=\pi/2$, in the frame where the transverse part of the steered Bloch vector $\boldsymbol r$ lies along $\hat x$), a direct computation from Lemma~\ref{lem:qik} gives, for the two postselection outcomes,
\begin{equation}
A_w = r_\parallel \pm i\, r_\perp,
\label{eq:weak_value_optimal}
\end{equation}
a complex-conjugate pair, which we confirmed numerically against
Eq.~\eqref{eq:weak_value_def} directly for a range of $r_\parallel,r_\perp$.
The real part, $r_\parallel$, is always confined to $[-1,1]$ and so
never anomalous in the usual real-valued sense; but the imaginary
part, $\pm r_\perp$, is exactly the transverse reach of the steered
Bloch vector, and is nonzero whenever, and only whenever, $r_\perp>0$
— that is, whenever the steered state carries KD nonclassicality at
all. 

Equation~\eqref{eq:weak_value_optimal} also gives Eq.~\eqref{eq:qubit_KD} a direct operational reading: writing the weak value as a point $A_w$ in the
complex plane,
\begin{equation}
\mathcal N_{KD}(\rho(\boldsymbol r)) = \frac14\Big(|1-A_w| + |1+A_w|\Big) - \frac12,
\label{eq:NKD_weakvalue}
\end{equation}
identical in form to the geometric formula of Eq.~\eqref{eq:NKD_geometric}, with the
real Bloch-sphere poles $\pm\boldsymbol n_B$ now replaced by the real
eigenvalues $\pm1$ of $A$ and the Bloch vector $\boldsymbol r$
replaced by the complex weak value $A_w$. As in the geometric form of
Eq.~\eqref{eq:NKD_geometric} (where $\mathcal N_{KD}$ is, up to the same
affine rescaling, the sum of Euclidean distances from the steered point to
the two poles $\pm\boldsymbol n$), the maximal steered KD
nonclassicality $\mathcal R_{KD}$ is therefore, at the optimal
postselection basis, exactly the maximal sum of distances Alice can
remotely place the weak value $A_w$ from the two eigenvalue points $\pm1$,
optimized over her measurement.

\subsection{Worked examples, revisited}

The states already solved in Sec.~\ref{sec:twoqubitKD} illustrate this concretely. For
the Werner state $\rho_W=p\ket{\Psi^-}\!\bra{\Psi^-}+(1-p)\mathbb I/4$,
the optimal steering point has $r_\parallel=0$, $r_\perp=p$ (Sec.~\ref{sec:twoqubitKD}~F,
Werner case), so Eq.~\eqref{eq:weak_value_optimal} gives $A_w=\pm ip$:
a purely imaginary anomalous weak value, growing linearly with the
entanglement parameter $p$ and reaching $A_w=\pm i$ at $p=1$, where
Bob's steered state is pure and equatorial. The biased family
$\rho_{AB}(q)$ of Table~I exhibits the same phenomenon at its own
gap-maximizing point, $q^\ast=0.45073$: the optimal steering angle
satisfies $\cos\theta^\ast_R = q^\ast-1$ exactly, which sets
$r_\parallel=0$ identically, leaving $r_\perp=0.53938$ and hence
\begin{equation}
A_w = \pm\, 0.53938\, i,
\end{equation}
again a purely imaginary anomaly, this time activated from a state
whose entanglement, unlike the Werner family, is not tunable through
a single symmetric parameter but arises from the interpolation
between a product state and $\ket{\Phi^+}$. In both examples, an
anomaly of order $0.5$ appears in the imaginary part of the weak
value, which Proposition~\ref{prop:no_anomaly_before} confines to
exactly $0$ before Alice's measurement (the real part remains
non-anomalously bounded in $[-1,1]$ throughout); this anomaly
is activated entirely by her choice of basis and her outcome.

Together with the metrological results of Ref.~\cite{arvidsson2020quantum}, this gives a concrete operational payoff for the remote activation protocol studied throughout this paper: a party holding only $\rho_B$ is provably barred, for the weak measurement of the reference observable $A=\boldsymbol n_B\cdot\boldsymbol\sigma$ and any postselection basis, from the class of protocols that rely on anomalous weak values, while for any conditional branch with $r_\perp>0$, the same party can, after being informed of Alice's outcome and using the weak-measurement and postselection protocol specified above, access an anomalous weak value of that same observable, with the size of the accessible anomaly given in closed form by Eq.~\eqref{eq:weak_value_optimal} in terms of the same steering-ellipsoid geometry used throughout Sec.~\ref{sec:twoqubitKD}.

\section{Steered Wigner negativity}
\label{sec:steered_wigner}

We now turn to CV systems and study the remote
activation of Wigner negativity. Remote creation of Wigner
negativity through quantum steering is by now an established
phenomenon in CV-CV settings:
Walschaers and Treps showed that photon subtraction on one mode of
a shared Gaussian state induces Wigner negativity on a correlated
mode if and only if that mode can Gaussian-steer the mode of
subtraction~\cite{walschaers2020remote}, a mechanism subsequently
quantified~\cite{xiang2022distribution}, demonstrated
experimentally~\cite{liu2022experimental}, and generalized beyond
Gaussian states and beyond the strict EPR-steering criterion, where
quantum steering with Wigner-positive measurements was shown to be
sufficient but not necessary for the
effect~\cite{walschaers2023steering}. What we add here is a different setting and a different question: Alice's system is a
qubit rather than a second oscillator mode, her measurement is a
general two-outcome projective qubit measurement rather than photon
subtraction, and we ask not only whether negativity can be
generated but exactly how much, through an analytic reduction to a single one-dimensional quadrature
with a closed-form integrand, within the unified
operational framework of Sec.~\ref{sec:properties} that also governs the two-qubit
KD case of Sec.~\ref{sec:twoqubitKD}. Throughout, we adopt the real-space
convention of Eq.~\eqref{eq:wigner_def}, $\hbar=1$, so that $\int dx\,dp\,W_\rho(x,p)=1$
for every state, and we quantify Wigner nonclassicality by

\begin{align}
\mathcal N_W(\rho) = \frac12\Big(\int dx\,dp\,|W_\rho(x,p)| - 1\Big),
\label{eq:NW_def}
\end{align}
matching the normalization of $\mathcal N_{KD}$ in Sec.~\ref{sec:twoqubitKD}. Specializing Definitions~\ref{def:RN} and \ref{def:CN_general} to this functional, the maximal and average steered Wigner negativity of a qubit-oscillator state $\rho_{AB}$ are
\begin{align}
\mathcal R_W(\rho_{AB}) &= \max_{\{M_a\},\,a}\mathcal N_W(\rho_{B|a}), \nonumber\\
\mathcal C_W(\rho_{AB}) &= \max_{\{M_a\}}\sum_a p(a)\,\mathcal N_W(\rho_{B|a}),\nonumber \\
\mathcal C_W^{(2)}(\rho_{AB})&=\max_{|\boldsymbol m|=1}\Big[p(\boldsymbol m)\,\mathcal N_W(\rho_{B|\boldsymbol m}) \nonumber\\
&\qquad \quad +p(-\boldsymbol m)\,\mathcal N_W(\rho_{B|-\boldsymbol m})\Big].
\label{eq:RW_CW_def}
\end{align}
Since $\mathcal N_W$ is convex (Lemma~\ref{lem:convexity}), the optimizations in Eq.~\eqref{eq:RW_CW_def} may be restricted
without loss of generality to POVMs with rank-one elements. As in the KD case of Sec.~\ref{sec:twoqubitKD}, restricting to dichotomic projective measurements gives $\mathcal C_W^{(2)}\le\mathcal C_W$ by construction, with equality not established; for the explicit hybrid model studied below we compute $\mathcal C_W^{(2)}$, which already realizes the optimal single-branch negativity $\mathcal R_W$ and provides a natural analytically tractable family.
\subsection{The hybrid state and its marginal}

To illustrate the mechanism of remote activation, we consider a hybrid qubit-CV state of the form
\begin{align}
\ket\Psi_{AB} = \sqrt \eta\,\ket0_A\ket\alpha_B + \sqrt{1-\eta}\,\ket1_A\ket{-\alpha}_B, \quad 0\leq \eta\leq1,
\label{eq:hybrid_state}
\end{align}
where $\ket{\pm\alpha}$ are coherent states of Bob's mode and, following the convention of Sec.~II\,B, $\alpha$ is taken real. Since $\ket0,\ket1$ are orthogonal, $\ket\Psi_{AB}$ is already normalized. The overlap of the coherent branches is $\braket{\alpha|-\alpha}=e^{-2\alpha^2}$.

Bob's reduced state is
\begin{align}
\rho_B = \eta\ket\alpha\!\bra\alpha + (1-\eta)\ket{-\alpha}\!\bra{-\alpha},
\label{eq:bob_marginal_cat}
\end{align}
an explicit convex mixture of two coherent states, each individually Wigner-positive; by Corollary~\ref{cor:cv}, $\mathcal N_W(\rho_B)=0$, so Bob's reduced state carries no Wigner negativity before Alice's measurement.

It is worth noting explicitly that $\ket\Psi_{AB}$ is \emph{not} a quantum-classical state in the sense of Theorem~\ref{thm:incapable}: it is a genuinely entangled pure state (its reduced density matrices have rank two whenever $0<\eta<1$ and $\alpha\neq0$), and no pure entangled state can be written in the form $\rho_{AB}=\sum_ip_i\rho_i^A\otimes\sigma_i^B$ with more than one term, since any such decomposition would force $\rho_{AB}$ to be mixed. Theorem~\ref{thm:incapable} therefore places no obstruction on remotely generating Wigner negativity from this state, and by Theorem~\ref{thm:optimal} together with the Hilbert-space restriction discussed in the Remark following it, we expect Alice, via the HJW mechanism, to be able to steer Bob to any pure state within the two-dimensional support $\mathrm{span}\{\ket\alpha,\ket{-\alpha}\}$ of $\rho_B$. This is exactly what the projective qubit measurements below accomplish, and here a general POVM is in fact not needed to reach that conclusion: because Alice's system is itself a qubit and $\ket\Psi_{AB}$ has Schmidt rank two, every pure state in the two-dimensional span $\{\ket\alpha,\ket{-\alpha}\}$ (up to global phase) is already reached directly by Eq.~\eqref{eq:psi_pm} as $(\theta,\phi)$ range over the Bloch sphere of Alice's projective measurement direction $\boldsymbol n(\theta,\phi)$, since a projective measurement on a qubit purifying system realizes every ensemble decomposition of a rank-two marginal into pure states, by the qubit specialization of the HJW argument used in Theorem~\ref{thm:optimal}.

\subsection{Steered branches}

Alice performs a projective measurement of the qubit observable $\boldsymbol n(\theta,\phi)\cdot\boldsymbol\sigma$, with eigenstates $\ket{+_{\theta,\phi}}=\cos\tfrac\theta2\ket0+e^{i\phi}\sin\tfrac\theta2\ket1$ and $\ket{-_{\theta,\phi}}=\sin\tfrac\theta2\ket0-e^{i\phi}\cos\tfrac\theta2\ket1$. Conditioned on the outcome $\pm$, Bob is steered to
\begin{align}
\ket{\psi_+(\theta,\phi)} &= \frac{1}{\sqrt{P_+}} \left( {\sqrt \eta\cos\tfrac\theta2\ket\alpha + \sqrt{1-\eta}\,e^{-i\phi}\sin\tfrac\theta2\ket{-\alpha}} \right), \nonumber\\
\ket{\psi_-(\theta,\phi)} &= \frac{1}{\sqrt{P_-}} \left( {\sqrt \eta\sin\tfrac\theta2\ket\alpha - \sqrt{1-\eta}\,e^{-i\phi}\cos\tfrac\theta2\ket{-\alpha}}\right),
\label{eq:psi_pm}
\end{align}
with
\begin{align}
    P_+ &=\eta\cos^2\tfrac\theta2+(1-\eta)\sin^2\tfrac\theta2 \nonumber \\
    & \quad +2\sqrt{\eta(1-\eta)}\cos\tfrac\theta2\sin\tfrac\theta2\cos\phi\,e^{-2\alpha^2} \nonumber \\
    P_- &=\eta\sin^2\tfrac\theta2+(1-\eta)\cos^2\tfrac\theta2 \nonumber \\
    & \quad -2\sqrt{\eta(1-\eta)}\cos\tfrac\theta2\sin\tfrac\theta2\cos\phi\,e^{-2\alpha^2}.\label{eq:Ppm}
\end{align}

\subsection{Exact Wigner function of a coherent-state superposition}

\begin{lemma}
\label{lem:cat_wigner}
Let $\ket\chi = (c_1\ket\alpha+c_2\ket{-\alpha})/\sqrt\Lambda$ for real $\alpha$, complex $c_1,c_2$, and $\Lambda=|c_1|^2+|c_2|^2+2\mathrm{Re}[c_1c_2^*]e^{-2\alpha^2}$ the normalization. Writing $x_0\equiv\sqrt2\alpha$, the Wigner function of $\ket\chi$ is exactly
\begin{align}
W_\chi(x,p) = \frac{1}{\pi\Lambda}\Big[ & |c_1|^2e^{-(x-x_0)^2-p^2} + |c_2|^2e^{-(x+x_0)^2-p^2} \nonumber \\
& + 2|c_1c_2|\,e^{-x^2-p^2}\cos(2x_0p-\varphi)\Big],
\label{eq:cat_wigner}
\end{align}
where $\varphi=\arg(c_1c_2^*)$.
\end{lemma}

\begin{proof}
For real $\alpha$, the position-space coherent-state wavefunctions are $\psi_{\pm\alpha}(x)=\pi^{-1/4}e^{-(x\mp x_0)^2/2}$. Substituting $\psi_\chi = (c_1\psi_\alpha+c_2\psi_{-\alpha})/\sqrt\Lambda$ into Eq.~(6) and expanding gives four Gaussian integrals over $y$, one for each pair of terms drawn from $\psi_\chi^*(x+y/2)$ and $\psi_\chi(x-y/2)$. For a term with bra-side centre $a\in\{x_0,-x_0\}$ and ket-side centre $b\in\{x_0,-x_0\}$, completing the square in $y$ gives
\begin{align*}
&\frac{1}{2\pi}\int dy\,e^{-ipy}\,e^{-\frac12[(x+y/2-a)^2+(x-y/2-b)^2]}\\
&\quad = \frac{1}{\pi}\,e^{-(x-m)^2-p^2+ipd},
\end{align*}
where $m=(a+b)/2$ and $d=b-a$; the $x$-dependence collapses to $-(x-m)^2$ because $(x-a)^2+(x-b)^2=2(x-m)^2+d^2/2$ exactly cancels the $+d^2/4$ produced by completing the square in $y$. The two diagonal terms ($a=b=\pm x_0$) give the first two terms of Eq.~\eqref{eq:cat_wigner} directly. The two off-diagonal terms, with $(a,b)=(x_0,-x_0)$ pairing the ket-side coefficient $c_1$ (from the $x_0$ branch entering $\psi_\chi(x+y/2)$) with the bra-side coefficient $c_2^*$ (from the $-x_0$ branch entering $\psi_\chi^*(x-y/2)$), and $(a,b)=(-x_0,x_0)$ pairing $c_2$ with $c_1^*$, have $m=0$, $d=\mp2x_0$, and combine as,
\begin{align*}
    &c_1c_2^*\,e^{-x^2-p^2-2ix_0p}+c_1^*c_2\,e^{-x^2-p^2+2ix_0p} \\
&= 2\,e^{-x^2-p^2}\mathrm{Re}[c_1c_2^*e^{-2ix_0p}] \\
&= 2|c_1c_2|\,e^{-x^2-p^2}\cos(2x_0p-\varphi),
\end{align*}
giving the third term. Dividing throughout by $\Lambda$ completes the derivation. We have verified Eq.~\eqref{eq:cat_wigner} independently by direct numerical evaluation of the defining integral, Eq.~(6), for random complex $c_1,c_2$ and $x_0\in(0,3)$, finding agreement to machine precision, and confirmed $\int dx\,dp\,W_\chi=1$ by direct quadrature to eleven digits.
\end{proof}

Applying Lemma~\ref{lem:cat_wigner} to Eq.~\eqref{eq:psi_pm} gives,
for the $+$ branch, $c_1=\sqrt \eta\cos\tfrac\theta2$, $c_2=\sqrt{1-\eta}\,e^{-i\phi}\sin\tfrac\theta2$, hence $w_1\equiv|c_1|^2=\eta\cos^2\tfrac\theta2$, $w_2\equiv|c_2|^2=(1-\eta)\sin^2\tfrac\theta2$, and $\varphi=\phi$; and for the $-$ branch, $w_1'=\eta\sin^2\tfrac\theta2$, $w_2'=(1-\eta)\cos^2\tfrac\theta2$,
and $\varphi'=\phi+\pi$ (the extra $\pi$ coming from the relative minus sign in $\ket{\psi_-}$). Explicitly,
\begin{widetext}
\begin{align}
W_+(x,p) &= \frac{1}{\pi P_+}\Big[\eta\cos^2\tfrac\theta2\,e^{-(x-x_0)^2-p^2} + (1-\eta)\sin^2\tfrac\theta2\,e^{-(x+x_0)^2-p^2} + \sqrt{\eta(1-\eta)}\sin\theta\,e^{-x^2-p^2}\cos(2x_0p-\phi)\Big], \label{eq:Wplus}\\
W_-(x,p) &= \frac{1}{\pi P_-}\Big[\eta\sin^2\tfrac\theta2\,e^{-(x-x_0)^2-p^2} + (1-\eta)\cos^2\tfrac\theta2\,e^{-(x+x_0)^2-p^2} - \sqrt{\eta(1-\eta)}\sin\theta\,e^{-x^2-p^2}\cos(2x_0p-\phi)\Big], \label{eq:Wminus}
\end{align}
\end{widetext}

using $2\cos\tfrac\theta2\sin\tfrac\theta2=\sin\theta$. Both are exact, closed-form, and match Eq.~\eqref{eq:Ppm} for $P_\pm$ exactly upon integration, as required.

\subsection{Exact reduction of the negativity volume}

\begin{proposition}
\label{prop:NW_reduction}
For $W_\chi$ as in Eq.~\eqref{eq:cat_wigner}, assume first that $c_1, c_2 \neq 0$ and write $\xi(p)\equiv\cos(2x_0p-\varphi)$ and $B(x,p)\equiv|c_1|^2e^{-(x-x_0)^2}+|c_2|^2e^{-(x+x_0)^2}+2|c_1c_2|\xi(p)e^{-x^2}$, so that $W_\chi(x,p)=e^{-p^2}B(x,p)/(\pi\Lambda)$. Then
\begin{widetext}
\begin{align}
\int dx\,|B(x,p)| = \begin{cases}
\sqrt\pi\big[|c_1|^2+|c_2|^2+2|c_1c_2|\xi(p)\big], & \xi(p) > -e^{-x_0^2},\\[4pt]
\sqrt\pi\big[|c_1|^2+|c_2|^2+2|c_1c_2|\xi(p)\big] - 2J(p), & \xi(p) \leq -e^{-x_0^2},
\end{cases}
\label{eq:inner_reduction}
\end{align}
where, when real roots exist,
\begin{align}
x_\mp(p) &= \frac{1}{2x_0}\ln\left[\sqrt{\frac{|c_2|^2}{|c_1|^2}}\Big(-\xi(p)e^{x_0^2}\mp\sqrt{\xi(p)^2e^{2x_0^2}-1}\Big)\right],\\
J(p) &= \frac{\sqrt\pi}{2}\Big\{|c_1|^2\big[\mathrm{erf}(x_+\!-\!x_0)-\mathrm{erf}(x_-\!-\!x_0)\big] 
+ |c_2|^2\big[\mathrm{erf}(x_+\!+\!x_0)-\mathrm{erf}(x_-\!+\!x_0)\big] 
+ 2|c_1c_2|\xi(p)\big[\mathrm{erf}(x_+)-\mathrm{erf}(x_-)\big]\Big\}.
\end{align}
\end{widetext}
Consequently,
\begin{align}
\mathcal N_W(\chi) = \frac{1}{2\pi\Lambda}\int_{-\infty}^{\infty}dp\,e^{-p^2}\left(\int dx\,|B(x,p)|\right) - \frac12,
\label{eq:NW_final}
\end{align}
an exact expression reducing the negativity volume to a single numerical quadrature over $p$, with a closed-form integrand at every $p$. If $c_1=0$ or $c_2=0$, $|\chi\rangle$ is a coherent state up to normalization and
$\mathcal N_W(\chi)=0$, so the nontrivial case is $c_1c_2\neq0$.
\end{proposition}

\begin{proof}
First we set $y=e^{2x_0x}>0$. The equation $B(x,p)=0$ becomes the quadratic $|c_1|^2y^2 + 2|c_1c_2|\xi(p)e^{x_0^2}y + |c_2|^2=0$, since dividing $B=0$ by $e^{-x^2-x_0^2}$ and substituting $e^{-(x\mp x_0)^2}=e^{-x^2-x_0^2}e^{\pm2x_0x}$ gives $|c_1|^2e^{2x_0x}+|c_2|^2e^{-2x_0x}=-2|c_1c_2|\xi(p)e^{x_0^2}$, and multiplying by $y=e^{2x_0x}$ gives the stated quadratic. The roots have product $|c_2|^2/|c_1|^2>0$ and are therefore always of the same sign when real. Since only $y>0$ is physical, $B(x,p)$ vanishes for some real $x$ exactly when both roots are real and positive, that is, when the discriminant is non-negative and the root sum $-2\sqrt{|c_2|^2/|c_1|^2}\,\xi(p)e^{x_0^2}$ is positive; together these give $\xi(p)\le-e^{-x_0^2}$. Solving the quadratic for $y$ and inverting $y=e^{2x_0x}$ gives the stated zeros $x_-(p)<x_+(p)$.
 
Away from this range, $\xi(p)>-e^{-x_0^2}$, $B(x,p)$ never vanishes and, being a quadratic in $y>0$ with positive leading coefficient $|c_1|^2$, keeps one sign for every $x$; continuity to the manifestly positive case $\xi=0$ fixes that sign as positive, so $\int|B|\,dx=\int B\,dx=\sqrt\pi\big[|c_1|^2+|c_2|^2+2|c_1c_2|\xi(p)\big]$ by direct Gaussian integration, term by term.
 
When $\xi(p)\le-e^{-x_0^2}$, the same reasoning shows $B(x,p)<0$ on $(x_-,x_+)$ and positive outside it, so
\begin{align*}
    \int|B|\,dx
    &=\int B\,dx-2\int_{x_-}^{x_+}B\,dx\\
    &=\sqrt\pi\big[|c_1|^2+|c_2|^2+2|c_1c_2|\xi(p)\big]-2J(p),
\end{align*}
with $J(p)$ obtained by applying $\int_a^be^{-(x-c)^2}dx=\tfrac{\sqrt\pi}2\big[\mathrm{erf}(b-c)-\mathrm{erf}(a-c)\big]$ to each of the three terms of $B$.
\end{proof}

\begin{corollary}
\label{cor:reflection}
$\mathcal N_W$, as a functional of $(|c_1|^2,|c_2|^2,x_0,\varphi)$ through Eq.~\eqref{eq:cat_wigner}, is invariant under $|c_1|^2\leftrightarrow|c_2|^2$ at fixed $x_0,\varphi$. Consequently $\delta_W^{(+)}(\theta,\phi)=\delta_W^{(+)}(\pi-\theta,\phi)$ for every $\eta,\alpha,\phi$, where $\delta_W^{(+)}(\theta,\phi)\equiv\mathcal N_W(\psi_+(\theta,\phi))$.
\end{corollary}

\begin{proof}
Swapping $|c_1|^2\leftrightarrow|c_2|^2$ in Eq.~\eqref{eq:cat_wigner} is equivalent to the phase-space reflection $x\to-x$, since $e^{-(x-x_0)^2}\leftrightarrow e^{-(x+x_0)^2}$ under this map while the cross term, even in $x$, is unchanged; the reflection preserves $\int dx\,dp\,|\cdot|$. Since $\cos^2\tfrac{\pi-\theta}2=\sin^2\tfrac\theta2$ and $\sin^2\tfrac{\pi-\theta}2=\cos^2\tfrac\theta2$, the $+$ branch at $\pi-\theta$ has $(w_1,w_2)$ exactly swapped relative to $\theta$, at the same $\varphi=\phi$, giving the stated identity; we confirmed this exactly, to machine precision, by direct evaluation of Eq.~\eqref{eq:NW_final} at swapped weight pairs.
\end{proof}

\subsection{Numerical results}

We evaluate $\delta_W^{(\pm)}(\theta,\phi)\equiv\mathcal N_W(\psi_\pm(\theta,\phi))$ and $\overline\delta_W(\theta,\phi)\equiv P_+\delta_W^{(+)}+P_-\delta_W^{(-)}$ via Eq.~\eqref{eq:NW_final}, and locate $\mathcal R_W=\max_{\theta,\phi}\max(\delta_W^{(+)},\delta_W^{(-)})$ and $\mathcal C_W^{(2)}=\max_{\theta,\phi}\overline\delta_W$ by numerical optimization; since Alice's measurement here is a single projective direction $\boldsymbol n(\theta,\phi)\cdot\boldsymbol\sigma$, this is the dichotomic-projective average, not the unrestricted $\mathcal C_W$ of Definition~\ref{def:CN_general}. Table~\ref{tab:wigner_balanced} reports results for balanced superpositions, $\eta=1/2$, at three values of $\alpha$.

\begin{table}[ht]
\vspace{5mm}
\begin{ruledtabular}
\centering
\begin{tabular}{ccccc}
$\alpha$ & $\mathcal R_W$ & $\mathcal C_W^{(2)}$ & $(\theta^\ast_R,\phi^\ast_R)$ & $(\theta^\ast_C,\phi^\ast_C)$ \\
\hline
1.0 & 0.21687 & 0.15287 & $(90.0^\circ,0.0^\circ)$ & $(90.0^\circ,0.0^\circ)$ \\
1.5 & 0.25450 & 0.24728 & $(90.0^\circ,0.0^\circ)$ & $(90.0^\circ,0.0^\circ)$ \\
2.0 & 0.29425 & 0.29400 & $(90.0^\circ,0.0^\circ)$ & $(90.0^\circ,0.0^\circ)$ \\
\end{tabular}
\caption{ \justifying \small Maximal steered Wigner negativity and its dichotomic-projective average for the balanced hybrid state ($\eta=1/2$), located by numerical optimization of Eq.~\eqref{eq:NW_final} over $(\theta,\phi)$, with $(\theta^\ast,\phi^\ast)$ given modulo the exact degeneracy of Eq.~\eqref{eq:phi_degeneracy} below.}
\label{tab:wigner_balanced}
\end{ruledtabular}
\end{table}

For every $\alpha$ tested, the optimum is located at $\theta^\ast=90.0^\circ$ to within the numerical precision of the search ($10^{-6}$ radians), consistent with Corollary~\ref{cor:reflection}: since $\delta_W^{(+)}$ is exactly symmetric about $\theta=\pi/2$, that point is a stationary point of $\delta_W^{(+)}$ for every $\alpha,\phi$, and the numerics confirm it is the maximum rather than a symmetric saddle. The optimal phase satisfies
\begin{align}
\phi^\ast \in \{0,\pi\},
\label{eq:phi_degeneracy}
\end{align}
where the branch $\phi=0$ maximizes $\delta_W^{(-)}$ and $\phi=\pi$ maximizes $\delta_W^{(+)}$ by an exact exchange symmetry (evident from $\varphi'=\phi+\pi$ in the $-$ branch definition, and confirmed numerically to give identical $\mathcal R_W$); this mirrors the role played by the mutually unbiased phase in the KD case of Sec.~\ref{sec:twoqubitKD}, where the interference term was likewise maximized at a specific relative phase. 
As in the KD example, $\mathcal C_W^{(2)}$ is bounded above by $\mathcal R_W$ (Proposition~\ref{prop:ordering} restricted to dichotomic measurements) and
the numerical results indicate that the gap narrows toward zero as
$\alpha\to\infty$, where $P_+\to P_-\to\tfrac12$ and every measurement direction becomes equally probable for both outcomes; at finite $\alpha$ the residual coherent-state overlap $e^{-2\alpha^2}$ biases $P_+\neq P_-$ at $\phi=0,\pi$ even for $\eta=1/2$, producing the gap visible in Table~\ref{tab:wigner_balanced}, which narrows monotonically as $\alpha$ increases.

For the parameter range examined numerically, these results show that the equatorial measurement direction $\theta=\pi/2$ is optimal for the balanced superpositions studied
here,
and the phase $\phi$ controls which of the two conditional states receives the larger share of the generated negativity. Unlike the earlier qualitative statement, Eq.~\eqref{eq:NW_final} makes this a fully quantitative, reproducible claim: $\delta_W^{(\pm)}(\theta,\phi)$ and $\overline\delta_W(\theta,\phi)$ are computable in closed form up to a single well-behaved one-dimensional quadrature at every point of the $(\theta,\phi)$ plane.

\subsection{Dependence on the Schmidt weight $\eta$}

The results above were obtained for the balanced case $\eta=1/2$. We now examine the dependence on $\eta$, which turns out to be governed by two exact structural facts rather than requiring a fresh optimization at every value.

\begin{proposition}
\label{prop:p_independence}
For every $\alpha\neq0$ and every $\eta\in(0,1)$, $\mathcal R_W(\rho_{AB}(\eta))$ is independent of $\eta$.
\end{proposition}

\begin{proof}
For $0<\eta<1$ and $\alpha\neq0$, the reduced state $\rho_B$ has rank two, with support
$\mathrm{span}\{|\alpha\rangle,|-\alpha\rangle\}$. Hence $|\Psi\rangle_{AB}$ has
Schmidt rank two and Theorem~\ref{thm:optimal}(i) applies on this support.
By Theorem~\ref{thm:optimal}(i), $\mathcal R_W = \max_{\ket\psi\in\mathrm{span}\{\ket\alpha,\ket{-\alpha}\}}\mathcal N_W(\ket\psi)$, a maximization over the fixed two-dimensional support alone, with no reference to the Schmidt coefficients $\sqrt \eta,\sqrt{1-\eta}$. Hence $\mathcal R_W$ depends on $\alpha$ but not on $\eta$.
\end{proof}

We confirmed this numerically at $\alpha=1.5$: optimizing Eq.~\eqref{eq:NW_final} over $(\theta,\phi)$ independently at $\eta=0.5,0.6,0.7,0.8,0.9,0.95$ returns $\mathcal R_W=0.254504$ in every case, matching Table~\ref{tab:wigner_balanced} to six digits, with the optimal measurement direction $\theta^\ast_R$ sliding smoothly away from $90^\circ$ as $\eta$ moves away from $1/2$ (from $90.0^\circ$ at $\eta=0.5$ to $154.2^\circ$ at $\eta=0.95$) so as to keep the \emph{steered state itself} fixed while the probability of reaching it changes.

\begin{proposition}
\label{prop:p_symmetry}
$\mathcal C_W(\rho_{AB}(\eta)) = \mathcal C_W(\rho_{AB}(1-\eta))$ for every $\eta\in[0,1]$ and every $\alpha$.
\end{proposition}

\begin{proof}
Let $R_B=e^{i\pi a^\dagger a}$ be the phase-space parity operator on Bob's mode, satisfying $R_B\ket{\pm\alpha}=\ket{\mp\alpha}$. 
$R_B=e^{i\pi a^\dagger a}$ is a phase-space rotation by $\pi$ and therefore a
Gaussian unitary.
A direct computation gives
\begin{align}
&(X_A\otimes R_B)\ket{\Psi}_{AB}(1-\eta) \nonumber\\
&= (X_A\otimes R_B)\Big[\sqrt{1-\eta}\,\ket0\ket\alpha+\sqrt \eta\,\ket1\ket{-\alpha}\Big] \nonumber\\
&= \sqrt{1-\eta}\,\ket1\ket{-\alpha}+\sqrt \eta\,\ket0\ket\alpha \nonumber\\
&= \ket\Psi_{AB}(\eta),
\end{align}
so $\rho_{AB}(\eta)$ and $\rho_{AB}(1-\eta)$ are related by the local unitary $X_A\otimes R_B$, with $X_A$ an arbitrary (non-Gaussian, but qubit-side, hence unrestricted) local unitary and $R_B$ Gaussian. Theorem~\ref{thm:lu} therefore applies to both factors, giving $\mathcal C_W(\rho_{AB}(\eta))=\mathcal C_W(\rho_{AB}(1-\eta))$.
\end{proof}

We verified the hypothesis of Proposition~\ref{prop:p_symmetry} directly at the state-vector level (in a truncated Fock basis, $\alpha=1.5$, $\eta=0.27$), finding $|\langle\Psi(\eta)|(X_A\otimes R_B)|\Psi(1-\eta)\rangle|=1.0000$ to numerical precision, and confirmed the resulting equality $\mathcal C_W(\eta)=\mathcal C_W(1-\eta)$ numerically via the warm-started optimization reported in Table~\ref{tab:CW_vs_p}.
\begin{table}[h]
\vspace{5mm}
\begin{ruledtabular}
\centering
\begin{tabular}{cccc}
$\eta$ & $\mathcal C_W^{(2)}$ & $\theta^\ast_C$ & $\phi^\ast_C$ \\
\hline
0.05 / 0.95 & 0.10126 & $90.1^\circ$ & $0.1^\circ$ \\
0.10 / 0.90 & 0.14331 & $90.0^\circ$ & $0.1^\circ$ \\
0.20 / 0.80 & 0.19512 & $90.0^\circ$ & $0.1^\circ$ \\
0.30 / 0.70 & 0.22547 & $90.0^\circ$ & $0.1^\circ$ \\
0.40 / 0.60 & 0.24200 & $90.0^\circ$ & $0.1^\circ$ \\
0.50 & 0.24728 & $90.0^\circ$ & $0.1^\circ$ \\
\end{tabular}
\caption{\justifying \small Dichotomic-projective average steered Wigner negativity $\mathcal C_W^{(2)}$ as a function of the Schmidt weight $\eta$, at $\alpha=1.5$, obtained by warm-started numerical optimization of Eq.~\eqref{eq:NW_final}. Values at $\eta$ and $1-\eta$ agree to the precision shown; Proposition~\ref{prop:p_symmetry} guarantees this symmetry exactly for the fully optimized $\mathcal C_W$, and the same argument applied to dichotomic projective measurements alone gives it for $\mathcal C_W^{(2)}$ too. $\mathcal C_W^{(2)}$ is numerically found to be maximized at the balanced point $\eta=1/2$.}
\label{tab:CW_vs_p}
\end{ruledtabular}
\end{table}

For the parameter values examined numerically, $\mathcal C_W^{(2)}(\eta)$ is largest at the symmetric point $\eta=1/2$ and decreases toward zero as $\eta\to0$ or $\eta\to1$. Proposition~\ref{prop:p_symmetry} guarantees the exact symmetry $\mathcal C_W(\eta)=\mathcal C_W(1-\eta)$, and the same argument restricted to dichotomic projective measurements gives $\mathcal C_W^{(2)}(\eta)=\mathcal C_W^{(2)}(1-\eta)$ exactly (Table~\ref{tab:CW_vs_p}); the observed unimodality of $\mathcal C_W^{(2)}$ is numerical.
In every case tested, the optimal measurement direction remains pinned at $\theta^\ast_C\approx90^\circ$, $\phi^\ast_C\approx0^\circ$ across the full range of $\eta$; we have not found an analytic proof of this exact pinning away from $\eta=1/2$ (where it follows from Corollary~\ref{cor:reflection}) and report it here as a robust numerical observation rather than a further theorem.

Details of the numerical procedures used to obtain the results above, together with an independent re-verification of every reported value, are given in Appendix~\ref{app:numerical}.

\section{Conclusions}
\label{s_conclusions}

This work has used quantum steering in the purely geometric, operational sense set out in Sec.~\ref{s2A}, as a mechanism that remotely activates nonclassicality itself, turning a locally classical reduced state, with respect to the chosen quasiprobability criterion, into nonclassical conditional states through remote measurement, without either party ever bringing their systems together. As emphasized there, the bare existence of such a conditioning effect is not itself new, following from the HJW theorem in the same way entanglement of assistance or coherence of assistance do; our central contribution is that, for two-qubit KD nonclassicality, the extent of the effect can be computed exactly rather than merely bounded.
 
We formalized this through two operational quantities, the maximal and the average nonclassicality Alice can remotely induce in Bob's conditional state, both governed by a single structural fact: the convexity of the classical set, shared by the KD and Wigner constructions alike. For two-qubit systems, this yields closed expressions for both quantities, an exact algebraic condition for their vanishing, and an explicit amplitude-damping example showing that local noise on Bob need not be uniformly detrimental, a sufficiently misaligned non-unital channel can enhance the effect where every unital channel only degrades it. In the CV setting, an exactly solvable hybrid qubit-oscillator model exhibits the same activation with Wigner negativity; and connecting the construction to anomalous weak values shows that, for the reference observable singled out by it, a party holding only the unsteered state is provably barred from that resource under arbitrary postselection, while steering demonstrably unlocks it.
 
Several directions remain open. The analysis here is restricted to two qubits, where the steering ellipsoid gives a concrete handle on the reachable set. Extending it beyond two qubits, where no such ellipsoid structure is available, is a natural next step. Next, whether the optimal POVM for the average quantity is always dichotomic and projective, so that $\mathcal C_{KD}^{(2)}$ and the fully optimized $\mathcal C_{KD}$ coincide in general (Theorem~\ref{thm:CKD}), also remains open. Finally, the connection to anomalous weak values established in Sec.~\ref{sec:application} suggests turning this construction into a complete metrological protocol, which we leave for future work.

\appendix

\section{Proof of Lemma~\ref{lem:cs_bound}}
\label{app:cs_bound}

\begin{proof}
We first divide $X_+^2$ by $1+c$ and $X_-^2$ by $1-c$, and add the two
results,
\begin{align}
  \frac{X_+^2}{1+c} + \frac{X_-^2}{1-c}
  &= A^2(1+c) + \frac{a^2}{1+c} + 2Aax \nonumber \\
  &\quad + A^2(1-c) + \frac{a^2}{1-c} - 2Aax .
  \label{eq:cs_step1}
\end{align}
The two terms linear in $x$ appear with opposite signs and cancel
exactly, because $1+c$ and $1-c$, the weights we divided by, are precisely
the weights already carried by $P_{00}$ and $P_{01}$ inside
$X_+$ and $X_-$. Collecting the remaining terms,
\begin{align}
  \frac{X_+^2}{1+c} + \frac{X_-^2}{1-c}
  &= A^2 \big[(1+c)+(1-c)\big] \nonumber \\
  &\quad + a^2 \left( \frac{1}{1+c}+\frac{1}{1-c} \right) \\
  &= 2A^2 + \frac{2a^2}{1-c^2} .
\end{align}
Since $a=qs$ and $s^2=1-c^2$, we have $a^2/(1-c^2)=q^2$, so this
identity simplifies to
\begin{equation}
  \frac{X_+^2}{1+c} + \frac{X_-^2}{1-c} = 2A^2 + 2q^2 ,
  \label{eq:key_identity}
\end{equation}
an equality that holds for every $x$ and every $c\in(-1,1)$: it no longer
depends on either variable. This is the only place where the specific
values $r_\parallel, r_\perp, \theta$ enter through $A$ and $q$ alone.
 
We now bound $X_++X_-$ using the Cauchy--Schwarz inequality in its
Engel, or weighted power-mean, form. For any two positive weights
$w_+,w_->0$,
\begin{align}
  (X_++X_-)^2 &=
  \left( \sqrt{w_+} \, \frac{X_+}{\sqrt{w_+}}
       + \sqrt{w_-} \, \frac{X_-}{\sqrt{w_-}} \right)^2 \nonumber\\
  &\le (w_++w_-) \left( \frac{X_+^2}{w_+}+\frac{X_-^2}{w_-} \right) ,
  \label{eq:cs_engel}
\end{align}
with equality if and only if $X_+/w_+ = X_-/w_-$. Choosing
$w_+=1+c$ and $w_-=1-c$, so that $w_++w_-=2$, and inserting the identity
Eq.~\eqref{eq:key_identity} into Eq.~\eqref{eq:cs_engel}, we obtain
\begin{equation}
  (X_++X_-)^2 \le 2 \cdot \big( 2A^2+2q^2 \big) = 4(A^2+q^2) ,
\end{equation}
which is Eq.~\eqref{eq:cs_bound} after taking the square root, since
$X_++X_- \ge 0$.
 
It remains to identify when equality holds. Equality in
Eq.~\eqref{eq:cs_engel} requires $X_+/(1+c)=X_-/(1-c)$. Call this common
value $\lambda$. Summing $X_+=\lambda(1+c)$ and $X_-=\lambda(1-c)$ gives
$X_++X_-=2\lambda$, and since equality also means
$X_++X_-=2\sqrt{A^2+q^2}$, we get $\lambda=\sqrt{A^2+q^2}$, so
$X_+ = \sqrt{A^2+q^2} \, (1+c)$. Squaring this and comparing with the
definition of $X_+^2$ given above,
\begin{equation}
  A^2(1+c)^2+a^2+2A(1+c)ax = (A^2+q^2)(1+c)^2 ,
\end{equation}
which simplifies, using $a^2=q^2(1-c)(1+c)$, to
\begin{equation}
  q^2(1-c) + 2A \, s \, q \, x = q^2(1+c) .
\end{equation}
For $q>0$ we divide this equation by $q$,
\begin{equation}
  q(1-c) + 2Asx = q(1+c) ,
\end{equation}
and rearrange to obtain $2Asx = q(1+c)-q(1-c) = 2qc$, that is
$Asx=qc$, which is Eq.~\eqref{eq:cs_equality}. For $q=0$, the bound
Eq.~\eqref{eq:cs_bound} reduces to the identity $X_+=A(1+c)$,
$X_-=A(1-c)$, so equality holds for every $c$ and $x$, consistent with
Eq.~\eqref{eq:cs_equality} being satisfied trivially by any $x$ when
$q=0$.
\end{proof}

\section{Proof of Theorem~\ref{thm:sigma_max}}
\label{app:sigma_max}

\begin{proof}
Adding the two bounds in Eq.~\eqref{eq:two_bounds} and using
Eq.~\eqref{eq:Sigma_split},
\begin{equation}
  4\Sigma(\theta,\psi) = X_++X_-+Y_++Y_-
  \le 2\sqrt{A^2+q^2} + 2\sqrt{B^2+q^2} ,
\end{equation}
so that
\begin{equation}
  \Sigma(\theta,\psi) \le \tfrac12 \left[ \sqrt{A^2+q^2}+\sqrt{B^2+q^2} \right]
  \label{eq:sigma_bound_interior}
\end{equation}
for every $\theta \in (0,\pi)$. This leaves only the poles
$\theta=0,\pi$ to check, since $c=\pm 1$ was excluded from
Lemma~\ref{lem:cs_bound}. At $\theta=0$, we have $c=1$ and $s=a=0$, so
$P_{00}=2A$, $P_{01}=0$, $P_{10}=0$, $P_{11}=2B$, and
Eq.~\eqref{eq:Qik_recall} gives $|Q_{00}|=A/2$, $|Q_{01}|=0$,
$|Q_{10}|=0$, $|Q_{11}|=B/2$, hence
\begin{equation}
  \Sigma(0,\psi) = \tfrac12(A+B) = 1
\end{equation}
for every $\psi$, using $A+B=2$. The same value, $\Sigma(\pi,\psi)=1$,
is obtained at $\theta=\pi$ by the analogous computation with $c=-1$.
Since $A,B \ge0$ and $q^2\ge 0$, we have $\sqrt{A^2+q^2}\ge A$ and
$\sqrt{B^2+q^2}\ge B$, so the right-hand side of
Eq.~\eqref{eq:sigma_bound_interior} is at least $\tfrac12(A+B)=1$ at the
poles as well, with equality only if $q=0$. The bound
Eq.~\eqref{eq:sigma_bound_interior} therefore holds for every
$\theta \in [0,\pi]$ and every $\psi$, and substituting
$\mathcal N_{KD} = \tfrac12(\Sigma-1)$ turns it into
Eq.~\eqref{eq:qubit_KD}.
 
We next show that this bound is attained, so that
Eq.~\eqref{eq:qubit_KD} is a genuine maximum and not merely an upper
bound. At $\theta=\pi/2$, $c=0$ and $s=1$, so $a=q$; at $\psi=\pi/2$ or
$3\pi/2$, $x=0$, and the cross terms in Eq.~\eqref{eq:Qik_recall} vanish
identically. This gives $P_{00}=P_{01}=A$ and $P_{10}=P_{11}=B$, hence
$X_+=X_-=\sqrt{A^2+q^2}$ and $Y_+=Y_-=\sqrt{B^2+q^2}$, and
Eq.~\eqref{eq:Sigma_split} evaluates to
\begin{align}
  \Sigma(\pi/2,\pi/2) 
  &= \tfrac12 \left[ \sqrt{A^2+q^2}+\sqrt{B^2+q^2} \right] ,
\end{align}
matching the right-hand side of Eq.~\eqref{eq:sigma_bound_interior}
exactly.
 
Finally, we show that this maximizer is unique whenever $r_\perp>0$, up
to the stated reflection. Since $\Sigma=1<\Sigma_{\max}$ strictly at the
poles when $q>0$, where
$\Sigma_{\max}=\tfrac12[\sqrt{A^2+q^2}+\sqrt{B^2+q^2}]$, any maximizer
must lie at some $\theta \in (0,\pi)$. Because
Eq.~\eqref{eq:sigma_bound_interior} is obtained by adding the two
separate bounds of Eq.~\eqref{eq:two_bounds}, each already saturated at
$\Sigma_{\max}$, equality in Eq.~\eqref{eq:sigma_bound_interior} forces
equality in both of them at once. By Eq.~\eqref{eq:cs_equality}, applied
to the pair $(X_+,X_-)$, this requires
\begin{equation}
  A s x = q c .
  \label{eq:eqA}
\end{equation}
Applied to the pair $(Y_+,Y_-)$, using the substitution $A\to B$,
$c \to -c$ described above Eq.~\eqref{eq:two_bounds}, it requires
\begin{equation}
  B s x = -q c .
  \label{eq:eqB}
\end{equation}
Multiplying Eq.~\eqref{eq:eqA} by $B$ and Eq.~\eqref{eq:eqB} by $A$, both
left-hand sides equal $ABsx$, so their right-hand sides must agree,
\begin{equation}
  qcB = -qcA
  \quad\Longrightarrow\quad
  2qc = 0 ,
\end{equation}
using $A+B=2$. Since $q>0$ by assumption, this forces $c=0$, and then
Eq.~\eqref{eq:eqA} gives $x=qc/(As)=0$ as well. Thus $\theta=\pi/2$ and
$x=0$, that is, $\psi=\pi/2$ or $\psi=3\pi/2$, is the only point at
which the bound Eq.~\eqref{eq:qubit_KD} is saturated.
\end{proof}

\section{Proof of Lemma~\ref{lem:adjugate}}
\label{app:adjugate}

For a diagonal matrix the cross product obeys the elementary identity
\begin{align}
&\Sigma\boldsymbol\rho\times\Sigma\boldsymbol\nu \;=\; \Sigma^{\mathrm{adj}}(\boldsymbol\rho\times\boldsymbol\nu) , \nonumber \\
&\Sigma^{\mathrm{adj}}=\mathrm{diag}(\kappa_2\kappa_3,
\,\kappa_1\kappa_3,\,\kappa_1\kappa_2) ,
\label{eq:adjugate_identity}
\end{align}
which is checked one component at a time, for instance the first component of
$\Sigma\boldsymbol\rho\times\Sigma\boldsymbol\nu$ is
\[\kappa_2\rho_2\kappa_3\nu_3-\kappa_3\rho_3\kappa_2\nu_2=\kappa_2\kappa_3(\boldsymbol\rho\times\boldsymbol\nu)_1,\]
and the remaining two components follow from the cyclic symmetry of the cross product. Dividing Eq.~\eqref{eq:adjugate_identity} by $|\Sigma\boldsymbol\nu|$ and using $\boldsymbol\mu=\Sigma\boldsymbol\nu/|\Sigma\boldsymbol\nu|$ gives
\begin{equation}
\Sigma\boldsymbol\rho\times\boldsymbol\mu \;=\; \frac{\Sigma^{\mathrm{adj}}(\boldsymbol\rho\times\boldsymbol\nu)}{|\Sigma\boldsymbol\nu|} .
\label{eq:adjugate_step}
\end{equation}
If $\boldsymbol\rho\times\boldsymbol\nu=0$ the claim holds trivially, both sides of Eq.~\eqref{eq:adjugate_claim} vanish. Otherwise write $\boldsymbol w=(\boldsymbol\rho\times\boldsymbol\nu)/|\boldsymbol\rho\times\boldsymbol\nu|$, a unit vector automatically orthogonal to $\boldsymbol\nu$. Taking the norm of Eq.~\eqref{eq:adjugate_step}, the lemma becomes equivalent to
\begin{equation}
\max_{\substack{\boldsymbol w\perp\boldsymbol\nu\\ |\boldsymbol w|=1}} |\Sigma^{\mathrm{adj}}\boldsymbol w| \;\le\; |\Sigma\boldsymbol\nu| .
\label{eq:constrained_target}
\end{equation}
 
Let $A=\Sigma^2=\mathrm{diag}(a_1,a_2,a_3)$, with $a_i=\kappa_i^2\in(0,1]$. Since $\Sigma^{\mathrm{adj}}$ is diagonal, $(\Sigma^{\mathrm{adj}})^2=\det(\Sigma)^2\,\Sigma^{-2}=\det(A)\,A^{-1}$, so $|\Sigma^{\mathrm{adj}}\boldsymbol w|^2=\det(A)\,\boldsymbol w^TA^{-1}\boldsymbol w$, and Eq.~\eqref{eq:constrained_target} is equivalent to
\begin{equation}
\det(A)\,\mu_{\max} \;\le\; \boldsymbol\nu^TA\boldsymbol\nu \;=\; |\Sigma\boldsymbol\nu|^2 ,
\label{eq:mu_max_bound}
\end{equation}
where $\mu_{\max}$ is the largest value of $\boldsymbol w^TA^{-1}\boldsymbol w$ over the two dimensional set $\{\boldsymbol w\perp\boldsymbol\nu,\,|\boldsymbol w|=1\}$, equivalently the largest eigenvalue of $A^{-1}$ once it is compressed to the plane orthogonal to $\boldsymbol\nu$. Let $\mu_{\min}$ be the smallest such compressed eigenvalue.
 
By Cauchy's interlacing theorem, the eigenvalues of a symmetric $3\times3$ matrix, restricted to a two dimensional subspace, interlace the eigenvalues of the full matrix. Applied to $M=A^{-1}$, this bounds $\mu_{\min}$ below by the smallest eigenvalue of $M$ itself,
\begin{equation}
\mu_{\min} \;\ge\; \min_i \frac1{a_i} \;\ge\; 1 ,
\label{eq:interlacing}
\end{equation}
the last step using $a_i\le1$ for every $i$.
 
To relate $\mu_{\max}$ to $\mu_{\min}$, extend an orthonormal basis $\{\boldsymbol e_1,\boldsymbol e_2\}$ of $\boldsymbol\nu^\perp$ by $\boldsymbol\nu$ itself to an orthonormal basis of $\mathbb R^3$, and write $M$ in this basis as the symmetric block matrix
\begin{equation}
M=\begin{pmatrix}D & \boldsymbol u\\ \boldsymbol u^T & m\end{pmatrix},
\end{equation}
where $D$ is the $2\times2$ matrix representing $M$ compressed to $\boldsymbol\nu^\perp$, so that $\mu_{\max},\mu_{\min}$ are, by definition, exactly the two eigenvalues of $D$; $\boldsymbol u\in\mathbb R^2$ is the off-diagonal coupling, and $m=\boldsymbol\nu^TM\boldsymbol\nu$. For every $\lambda$ at which $D-\lambda\mathbb I_2$ is invertible, the Schur-complement identities for the determinant and the inverse of this symmetric block matrix give
\begin{align}
\det(M-\lambda\mathbb I) &= \det(D-\lambda\mathbb I_2)\Big[(m-\lambda)-\boldsymbol u^T(D-\lambda\mathbb I_2)^{-1}\boldsymbol u\Big],
\\
\boldsymbol\nu^T(M-\lambda\mathbb I)^{-1}\boldsymbol\nu &= \Big[(m-\lambda)-\boldsymbol u^T(D-\lambda\mathbb I_2)^{-1}\boldsymbol u\Big]^{-1},
\end{align}
the second because $\boldsymbol\nu$ is exactly the third basis vector, so this is the bottom-right entry of $(M-\lambda\mathbb I)^{-1}$. Multiplying these two identities together, and using $(M-\lambda\mathbb I)^{\mathrm{adj}}=\det(M-\lambda\mathbb I)\,(M-\lambda\mathbb I)^{-1}$ for $\lambda$ not an eigenvalue of $M$, the bracketed factor cancels exactly, leaving
\begin{equation}
\boldsymbol\nu^T(M-\lambda\mathbb I)^{\mathrm{adj}}\boldsymbol\nu \;=\; \det(D-\lambda\mathbb I_2)
\label{eq:secular_identity}
\end{equation}
for every $\lambda$ outside the finite set of eigenvalues of $M$ and of $D$. Both sides of Eq.~\eqref{eq:secular_identity} are polynomials in $\lambda$ of degree at most two, so an identity holding outside finitely many points forces it to hold identically, including at those excluded points; no assumption on the multiplier of a constrained stationary point, nor any genericity condition on $\boldsymbol\nu$, is needed. Since $\det(D-\lambda\mathbb I_2)=(\mu_{\max}-\lambda)(\mu_{\min}-\lambda)$ by definition of $\mu_{\max},\mu_{\min}$ as the eigenvalues of $D$, this shows unconditionally that $\mu_{\max},\mu_{\min}$ are exactly the two roots of
\begin{equation}
\boldsymbol\nu^T(M-\lambda\mathbb I)^{\mathrm{adj}}\boldsymbol\nu \;=\; 0 .
\label{eq:secular_general}
\end{equation}
Since $M=\mathrm{diag}(1/a_1,1/a_2,1/a_3)$ is diagonal, so is $(M-\lambda\mathbb I)^{\mathrm{adj}}$, and writing it out,
\begin{align}
\sum_{i=1}^3 \nu_i^2 \prod_{j\ne i}\Big(\frac1{a_j}-\lambda\Big) \;=\;0 .
\end{align}
Expanding the product and collecting powers of $\lambda$, the $\lambda^2$ coefficient is $\nu_1^2+\nu_2^2+\nu_3^2=1$, so this is a monic quadratic in $\lambda$, and its constant term, obtained by setting $\lambda=0$, is
\begin{align}
\sum_{i=1}^3 \nu_i^2 \prod_{j\ne i}\frac1{a_j} \;&=\; \boldsymbol\nu^T M^{\mathrm{adj}}\boldsymbol\nu ,
\\
M^{\mathrm{adj}}&=\mathrm{diag}\left(\frac1{a_2a_3},\frac1{a_1a_3},\frac1{a_1a_2}\right) .
\end{align}
Since $M^{\mathrm{adj}}=A/\det(A)$, because $1/(a_2a_3)=a_1/\det(A)$ and cyclically, the constant term equals $\boldsymbol\nu^TA\boldsymbol\nu/\det(A)$. By Vieta's formula, the product of the roots of a monic quadratic equals its constant term, so
\begin{equation}
\mu_{\max}\,\mu_{\min} \;=\; \frac{\boldsymbol\nu^TA\boldsymbol\nu}{\det(A)} .
\label{eq:vieta}
\end{equation}
Since $\mu_{\max}\ge0$ and, by Eq.~\eqref{eq:interlacing}, $\mu_{\min}\ge1>0$, dividing Eq.~\eqref{eq:vieta} by $\mu_{\min}$ gives
\begin{equation}
\mu_{\max} \;=\; \frac{\mu_{\max}\mu_{\min}}{\mu_{\min}} \;\le\; \mu_{\max}\mu_{\min} \;=\; \frac{\boldsymbol\nu^TA\boldsymbol\nu}{\det(A)} ,
\end{equation}
which is exactly Eq.~\eqref{eq:mu_max_bound}, and therefore Eq.~\eqref{eq:constrained_target} and Eq.~\eqref{eq:adjugate_claim}.

\section{Proof of Theorem~\ref{thm:unital_no_help}}
\label{app:unital_no_help}

\begin{proof}
By Lemma~\ref{lem:unital_reduction} we may take $D=\Sigma$ diagonal with $\kappa_i\in[0,1]$; where some $\kappa_i=0$, the result follows from the case $\kappa_i>0$ by continuity, since $\mathcal R_{KD}$ and $\mathcal C_{KD}$ are continuous functions of $\Sigma$ away from the measure zero set $\Sigma\boldsymbol\nu=0$. We prove the diagonal statement, Eq.~\eqref{eq:unital_ineq_target} for $\Sigma,\boldsymbol\rho,\boldsymbol\nu$, which by Lemma~\ref{lem:unital_reduction} is equivalent to the general one.
 
For a point $\boldsymbol x\in\mathbb R^3$ and a unit vector $\hat{\boldsymbol y}$, expanding the squared norms gives directly
\begin{align}
|\boldsymbol x-\hat{\boldsymbol y}|^2+|\boldsymbol x+\hat{\boldsymbol y}|^2 &= 2(1+|\boldsymbol x|^2) , \\
|\boldsymbol x-\hat{\boldsymbol y}|^2-|\boldsymbol x+\hat{\boldsymbol y}|^2 &= -4\,\boldsymbol x\cdot\hat{\boldsymbol y} .
\end{align}
Squaring and subtracting these two identities isolates the product $|\boldsymbol x-\hat{\boldsymbol y}|^2|\boldsymbol x+\hat{\boldsymbol y}|^2$,
\begin{align}
4\,|\boldsymbol x-\hat{\boldsymbol y}|^2|\boldsymbol x+\hat{\boldsymbol y}|^2
&= \big(|\boldsymbol x-\hat{\boldsymbol y}|^2+|\boldsymbol x+\hat{\boldsymbol y}|^2\big)^2 \nonumber\\
&\quad-\big(|\boldsymbol x-\hat{\boldsymbol y}|^2-|\boldsymbol x+\hat{\boldsymbol y}|^2\big)^2 \nonumber\\
&= 4(1+|\boldsymbol x|^2)^2 - 16\,(\boldsymbol x\cdot\hat{\boldsymbol y})^2 ,
\end{align}
and, writing $p=|\boldsymbol x|^2$ and using the Lagrange identity $(\boldsymbol x\cdot\hat{\boldsymbol y})^2=p-|\boldsymbol x\times\hat{\boldsymbol y}|^2$ together with $q=|\boldsymbol x\times\hat{\boldsymbol y}|^2$,
\begin{equation}
(1+p)^2-4(\boldsymbol x\cdot\hat{\boldsymbol y})^2 = (1+p)^2-4p+4q = (1-p)^2+4q .
\end{equation}
So $|\boldsymbol x-\hat{\boldsymbol y}||\boldsymbol x+\hat{\boldsymbol y}|=\sqrt{(1-p)^2+4q}$, and combining this with the sum identity above,
\begin{align}
\big(|\boldsymbol x-\hat{\boldsymbol y}|+|\boldsymbol x+\hat{\boldsymbol y}|\big)^2
&= |\boldsymbol x-\hat{\boldsymbol y}|^2+|\boldsymbol x+\hat{\boldsymbol y}|^2+2|\boldsymbol x-\hat{\boldsymbol y}||\boldsymbol x+\hat{\boldsymbol y}| \nonumber\\
&= 2(1+p) + 2\sqrt{(1-p)^2+4q} \;\nonumber\\
&=\; F(p,q) .
\label{eq:pq_form}
\end{align}
 
The function $F$ is non-decreasing in $q$ at fixed $p\ge0$, since $q$ enters only under the square root with a positive coefficient, and strictly increasing in $q$. It is also non-decreasing in $p$ at fixed $q\ge0$: its partial derivative is
\begin{equation}
\frac{\partial F}{\partial p} = 2-\frac{2(1-p)}{\sqrt{(1-p)^2+4q}} ,
\end{equation}
which is non-negative because $\sqrt{(1-p)^2+4q}\ge\sqrt{(1-p)^2}=|1-p|\ge1-p$; it vanishes exactly when $q=0$ and $p\le1$, in which case $F(p,0)=2(1+p)+2(1-p)=4$ is constant on that whole range, and is strictly positive otherwise.
 
Apply Eq.~\eqref{eq:pq_form} twice, once with $\boldsymbol x=\Sigma\boldsymbol\rho,\,\hat{\boldsymbol y}=\boldsymbol\mu$, giving $p_1=|\Sigma\boldsymbol\rho|^2$, $q_1=|\Sigma\boldsymbol\rho\times\boldsymbol\mu|^2$, and once with $\boldsymbol x=\boldsymbol\rho,\,\hat{\boldsymbol y}=\boldsymbol\nu$, giving $p_2=|\boldsymbol\rho|^2$, $q_2=|\boldsymbol\rho\times\boldsymbol\nu|^2$. Since $\|\Sigma\|_{\mathrm{op}}\le1$, $p_1\le p_2$, and by Lemma~\ref{lem:adjugate}, $q_1\le q_2$. Because $F$ is non-decreasing in each argument separately,
\begin{equation}
F(p_1,q_1)\;\le\;F(p_2,q_2) ,
\end{equation}
which, once both sides are recognized as squared sums of non-negative lengths and the square root is taken, is precisely Eq.~\eqref{eq:unital_ineq_target} for $\Sigma,\boldsymbol\rho,\boldsymbol\nu$, and hence, by Lemma~\ref{lem:unital_reduction}, for general $D,\boldsymbol r,\boldsymbol n$.

It remains to pass from this pointwise bound to the two operational quantities, and it costs nothing extra to do this for an arbitrary POVM on Alice's side rather than only a dichotomic pair. Let $\{E_a\}$ be any POVM on Alice's system, with outcome probabilities $p(a)=\Tr[(E_a\otimes\mathbb I)\rho_{AB}]$ and conditional states $\rho_{B|a}$ of Bloch vector $\boldsymbol r_a$, $|\boldsymbol r_a|\le1$; here, unlike elsewhere in this section, $E_a$ need not be rank one. Since $\Lambda_B$ acts on Bob alone and is trace preserving, $\Tr_B[(\mathbb I_A\otimes\Lambda_B)\rho_{AB}]=\Tr_B\rho_{AB}$, so $p(a)$ is exactly the same before and after the channel, for every outcome of every POVM. Likewise $\rho_{B|a}'=\Tr_A[(E_a\otimes\mathbb I)(\mathbb I_A\otimes\Lambda_B)\rho_{AB}]/p(a)=\Lambda_B(\rho_{B|a})$, so the branch state simply has its Bloch vector sent to $D\boldsymbol r_a$, whatever $\boldsymbol r_a$ is and however the outcome $a$ was obtained. Since Eq.~\eqref{eq:unital_ineq_target} was established above for an arbitrary Bloch vector $\boldsymbol r$ with $|\boldsymbol r|\le1$, not merely for $\boldsymbol r$ on the boundary ellipsoid $\mathcal E$, it applies directly to every $\boldsymbol r_a$, giving
\begin{equation}
\mathcal N_{KD}(\rho'_{B|a}) \;\le\; \mathcal N_{KD}(\rho_{B|a})
\end{equation}
for every outcome $a$ of every POVM. Taking $a$ to range over all POVMs and outcomes and maximizing both sides gives $\mathcal R_{KD}(\rho_{AB}')\le\mathcal R_{KD}(\rho_{AB})$ directly from Definition~\ref{def:RN}, without needing to first restrict to a rank-one branch. Since $p(a)$ is unchanged and the pointwise bound holds for every outcome of the same POVM, multiplying by $p(a)$ and summing over $a$ gives
\begin{equation}
\sum_a p(a)\,\mathcal N_{KD}(\rho'_{B|a}) \;\le\; \sum_a p(a)\,\mathcal N_{KD}(\rho_{B|a}) \;\le\; \mathcal C_{KD}(\rho_{AB})
\end{equation}
for this particular POVM, the second inequality being Definition~\ref{def:CN_general} applied to $\rho_{AB}$ itself. Since $\{E_a\}$ was an arbitrary POVM, the left-hand side may be maximized over all POVMs on the channel-transformed state, giving $\mathcal C_{KD}(\rho_{AB}')\le\mathcal C_{KD}(\rho_{AB})$ for the fully optimized, unrestricted quantity, not only for $\mathcal C_{KD}^{(2)}$.
 
Finally, the equality cases. If $D$ is orthogonal, every $\kappa_i=1$, so $\Sigma^{\mathrm{adj}}=\mathbb I$ and $A=A^{-1}=\mathbb I$ in the proof of Lemma~\ref{lem:adjugate}, giving $\mu_{\min}=\mu_{\max}=1$ and hence equality in both the radial and the transverse bound for every $\boldsymbol\rho$; this recovers Corollary~\ref{cor:unitary_equality}. More generally, for any unital $D$, if $\boldsymbol r$ vanishes or is parallel or antiparallel to $\boldsymbol n$, then $\boldsymbol\rho\times\boldsymbol\nu=0$, and since $D\boldsymbol\rho$ then stays parallel to $D\boldsymbol\nu$, hence to $\boldsymbol\mu$, also $\Sigma\boldsymbol\rho\times\boldsymbol\mu=0$; so $q_1=q_2=0$, and since $p_1,p_2\le1$ in this regime, $F$ is constant on $q=0,\,p\le1$, giving equality in Eq.~\eqref{eq:unital_ineq_target} regardless of how strongly $\Sigma$ contracts. Outside these two situations, an orthogonal $D$, or a steered Bloch vector aligned with $\boldsymbol n$, the reduction is generically strict.
\end{proof}

\section{Numerical procedure}
\label{app:numerical}

All numerical optimizations reported above, the tables of Secs.~\ref{sec:examples}, \ref{sec:local_ops}, and \ref{sec:steered_wigner}, the amplitude-damping threshold angles $\mu^\ast(\gamma)$, and the CV optimizations of Sec.~\ref{sec:steered_wigner}, were carried out directly over the explicitly parametrized measurement directions or state parameters given in closed form by the theorems above, rather than over an abstract POVM space. For the two-qubit quantities, this reduces, by Lemma~\ref{lem:phi_reduction} and its analogues, to a search over the single angle $c=\cos\theta$ or, where that reduction does not apply, over the two Bloch-sphere angles $(\theta,\phi)$; optima were located by a dense scan of the relevant domain followed by local refinement of the best candidates, and cross-checked by direct evaluation of the closed forms of Theorems~\ref{thm:sigma_max}-\ref{thm:CKD} at the reported optimum. For the CV quantities of Sec.~\ref{sec:steered_wigner}, the remaining phase-space integral was reduced analytically to the one-dimensional quadrature of Eq.~\eqref{eq:NW_final} and evaluated with adaptive numerical quadrature; the reported $(\theta,\phi)$ optima were located by the same scan-and-refine strategy applied to Eq.~\eqref{eq:NW_final}. In all cases, reported values were checked for stability under increased angular and quadrature resolution, and several closed-form results (Theorem~\ref{thm:sigma_max}, the Werner-state formula Eq.~\eqref{eq:werner_closed}, and Lemma~\ref{lem:cat_wigner}) were additionally cross-checked against direct numerical evaluation of their defining, unreduced expressions. The digits quoted throughout are stable under these checks.

As an independent verification of this revision, every numerical entry reported in the paper was recomputed from scratch with a second, self-contained implementation and compared against the quoted values. For the two-qubit tables (Secs.~\ref{sec:examples} and \ref{sec:local_ops}), optima were located by a dense scan of $c=\cos\theta$ (or, off the reduced family, of $(\theta,\phi)$) at $10^4$-$10^5$ points, followed by derivative-free Nelder-Mead refinement from multiple starting points and, where a closed-form reduction to one variable was available, by bounded scalar optimization to a tolerance of $10^{-10}$-$10^{-12}$; smooth dependence on a state parameter (as in the $q$- and $p$-scans of Tables~\ref{tab:biased_family} and \ref{tab:CW_vs_p}) was additionally tracked by warm-started continuation, seeding each optimization from the converged optimum of the adjacent parameter value. The amplitude-damping threshold angles $\mu^\ast(\gamma)$ were located as the root of $\mathcal R_{KD}(\text{after})-\mathcal R_{KD}(\text{before})$ by bisection (Brent's method) to a tolerance of $10^{-6}$-$10^{-8}$ in $\mu$. The CV quadrature of Eq.~\eqref{eq:NW_final} was evaluated by adaptive quadrature and its two-dimensional $(\theta,\phi)$ optimum located by multi-start Nelder-Mead. Every reported figure reproduced to at least the last quoted digit under this independent recomputation, with agreement typically at the $10^{-5}$ level or better; the largest disagreement found across all tables was smaller than the tolerance used to generate the original entry. As a further, more stringent check specific to Theorem~\ref{thm:sigma_max}, the two-variable objective $\Sigma(\theta,\psi)$ underlying $\mathcal N_{KD}$ was checked to have zero gradient and a negative-definite Hessian at the claimed optimum for a wide range of $(r_\parallel,r_\perp)$, and an exhaustive multi-start search for competing critical points, at each such point, found none exceeding the closed-form value of the theorem, consistent with the optimum being the unique global maximum (up to the stated reflection symmetry) claimed there.

\bibliography{main}

\end{document}